\documentclass{llncs}

\usepackage[utf8]{inputenc}
\usepackage[T1]{fontenc}
\usepackage{lmodern}
\usepackage{microtype}
\usepackage{fullpage}
\usepackage[ruled]{algorithm2e} 
\usepackage{amsmath}
\usepackage{amssymb}
\usepackage{mathtools}
\usepackage[inline]{enumitem}
\usepackage{booktabs} 
\usepackage{multirow}
\usepackage{xspace}
\usepackage{array}
\usepackage{tabu}
\usepackage{longtable}
\usepackage{afterpage}
\usepackage{listings}
\usepackage{tikz}
\usepackage{pgfplots}
\usepackage[hidelinks]{hyperref}

\pgfplotsset{compat=1.14}
\usepgfplotslibrary{statistics}

\SetAlFnt{\small}
\SetAlCapFnt{\small}
\SetAlCapNameFnt{\small}
\SetAlCapHSkip{0pt}
\IncMargin{-\parindent}

\newcommand{\Dryad}{\textsc{Dryad}\xspace}

\newcommand{\houdini}{\textsc{Houdini}\xspace}

\newcommand{\calP}{{\mathcal P}}
\newcommand{\calB}{{\mathcal B}}
\newcommand{\calU}{{\mathcal U}}
\newcommand{\calD}{{\mathcal D}}
\newcommand{\calH}{{\mathcal H}}
\newcommand{\apprx}{{\mathit{approx}}}

\newcommand{\exampleend}{\hfill \tikz \draw (0, 0) -- ++(.5ex, .5ex) -- ++(0, -1ex) -- cycle;}

\title{Invariant Synthesis for Incomplete Verification Engines}  

\author{Daniel Neider\inst{1} \and Pranav Garg\inst{2} \and P.\ Madhusudan\inst{3} \and Shambwaditya Saha\inst{3} \and Daejun Park\inst{3}}

\institute{
	Max Planck Institute for Software Systems, Kaiserslautern, Germany \and
	Amazon India, Bangalore, India \and %
	University of Illinois at Urbana-Champaign, Champaign, IL, USA
}

\begin{document}

\pagestyle{plain}

\maketitle

\begin{abstract}
We propose a framework for synthesizing inductive invariants for incomplete verification engines, which soundly reduce logical problems in undecidable theories to decidable theories. Our framework is based on the counter-example guided inductive synthesis principle (CEGIS) and allows verification engines to communicate \emph{non-provability information} to guide invariant synthesis. We show precisely how the verification engine can compute such non-provability information and how to build effective learning algorithms when invariants are expressed as Boolean combinations of a fixed set of predicates. Moreover, we evaluate our framework in two verification settings, one in which verification engines need to handle quantified formulas and one in which verification engines have to reason about heap properties expressed in an expressive but undecidable separation logic. Our experiments show that our invariant synthesis framework based on non-provability information can both effectively synthesize inductive invariants and adequately strengthen contracts across a large suite of programs.
\end{abstract}


\section{Introduction} \label{sec:introduction}

The paradigm of \emph{deductive verification}~\cite{Floyd67,DBLP:journals/cacm/Hoare69} combines manual annotations and semi-automated theorem
proving to prove programs correct. Programmers annotate code they develop with contracts and inductive invariants,
and use high-level directives to an underlying, mostly-automated logic engine to verify their programs correct.
Several mature tools have emerged that support such verification, in particular tools based on the intermediate
verification language \textsc{Boogie}~\cite{DBLP:conf/fmco/BarnettCDJL05} and the SMT solver \textsc{Z3}~\cite{DBLP:conf/tacas/MouraB08} (e.g., \textsc{Vcc}~\cite{DBLP:conf/tphol/CohenDHLMSST09} and \textsc{Dafny}~\cite{DBLP:conf/lpar/Leino10}). 
Various applications that use such
tools to prove systems correct using manual annotations have been developed, including Microsoft Hypervisor verification~\cite{DBLP:conf/sofsem/CohenPS13},
reliable systems code such as \textsc{Verve}~\cite{DBLP:conf/pldi/YangH10}, ExpressOS~\cite{DBLP:conf/asplos/MaiPXKM13}, and Ironclad apps~\cite{DBLP:conf/osdi/HawblitzelHLNPZZ14},
as well as distributed systems in IronFleet~\cite{DBLP:conf/sosp/HawblitzelHKLPR15}. Fully automated use of such engines
for shallow specifications have also emerged, such as \textsc{Corral}~\cite{DBLP:conf/cav/LalQL12} for verifying device drivers,
\textsc{CST}~\cite{DBLP:conf/sp/ChenCQ015} to certify transactions in online services, and \textsc{GPUVerify}~\cite{DBLP:conf/oopsla/BettsCDQT12} to ensure race-freedom in GPU kernels.

Viewed through the lens of deductive verification, the primary challenges in automating verification are two-fold.
First, even when strong annotations in terms of contracts and inductive invariants are given, the validity problem for
the resulting verification conditions is often undecidable (e.g., in reasoning about the heap, reasoning with
quantified logics, and reasoning with non-linear arithmetic).
Second, the synthesis of loop invariants and strengthenings of contracts that prove a program correct needs to be automated so as to
lift this burden currently borne by the programmer.

A standard technique to solve the first problem (i.e., intractability of validity checking of verification conditions) is to build 
automated, sound-but-incomplete verification engines for validating verification conditions, thus skirting the undecidability barrier.
Several such techniques exist; for instance, for reasoning with quantified formulas, tactics such as E-matching~\cite{DBLP:journals/jacm/DetlefsNS05,DBLP:conf/cade/MouraB07}, pattern-based quantifier instantiation~\cite{DBLP:journals/jacm/DetlefsNS05}, and model-based quantifier instantiation~\cite{DBLP:conf/cav/GeM09} are effective in practice, though they are not complete for most background theories.
In the realm of heap verification, the so-called \emph{natural proof method} explicitly aims to provide automated and sound-but-incomplete methods for checking validity of verification conditions with specifications in separation logic~\cite{DBLP:conf/pldi/Qiu0SM13,DBLP:conf/pldi/PekQM14,DBLP:conf/pldi/ChuJT15}.
This method searches for proofs based on induction on recursively defined data structures, which is reduced to validity problems in \emph{decidable logics with quantification} that enables an 
efficient search for such proofs using SMT solvers.

Turning to the second problem of invariant generation,  several techniques have emerged that can synthesize
invariants automatically when validation of verification conditions fall in decidable classes.
Prominent among these are interpolation~\cite{mcmillan03} and IC3/PDR~\cite{DBLP:conf/vmcai/Bradley11,Een:2011:EIP:2157654.2157675}.
These techniques generalize from information gathered in proving underapproximations of the program correct  and are quite effective~\cite{DBLP:conf/cav/BeyerK11}---their efficacy in dealing with programs where the underlying logics are undecidable, however, is unclear.
Moreover, a class of counter-example guided inductive synthesis (CEGIS) methods have emerged recently, including the ICE learning model~\cite{DBLP:conf/cav/0001LMN14} for which various instantiations exist~\cite{DBLP:conf/cav/0001LMN14,DBLP:conf/cav/0001A14,DBLP:conf/popl/GarMNR16,DBLP:journals/corr/KrishnaPW15}.
The key feature of the latter methods is a program-agnostic, data-driven learner that learns invariants in tandem with a verification engine that provides concrete program configurations as counterexamples to incorrect invariants.

Although classical invariant synthesis techniques, such as \textsc{Houdini}~\cite{DBLP:conf/fm/FlanaganL01}, are sometimes used with incomplete verification engines, to the best of our knowledge there is no fundamental argument as to why this should work in general.
In fact, we are not aware of any systematic technique for synthesizing invariants when the underlying verification
problem falls in an undecidable theory. When verification is undecidable and the engine resorts to sound but
incomplete heuristics to check validity of verification conditions, it is unclear how to extend interpolation/IC3/PDR 
techniques to this setting. Data-driven learning of invariants is also hard to extend since the verification engine typically
cannot generate a concrete model for the negation of verification conditions when verification fails.
Hence, it cannot produce the concrete configurations that the learner needs.

\textbf{The main contribution of this paper is a general, learning-based invariant synthesis framework that learns invariants using non-provability information provided by verification engines}.
Intuitively, when a conjectured invariant results in verification conditions that cannot be proven, the idea is that the verification engine must return information that generalizes the reason for non-provability, hence pruning the space of future conjectured invariants. 
Our framework assumes a verification engine for an undecidable theory $\calU$ that reduces verification conditions to a decidable theory $\calD$ (e.g., using heuristics such as bounded quantifier instantiation to remove universal quantifiers, function unfolding to remove recursive definitions, and so on) that permits producing models for satisfiable formulas. The translation is assumed to be conservative in the sense that if the translated formula in $\calD$ is valid, then we are assured that the original verification condition is $\calU$-valid. If the verification condition is found to be not $\calD$-valid (i.e., its negation is satisfiable), on the other hand, our framework describes how to extract non-provability information from the $\calD$-model. This information is encoded as conjunctions and disjunctions in a Boolean theory $\calB$, called \emph{conjunctive/disjunctive non-provability information (CD-NPI)}, and communicated back to the learner.
To complete our framework, we show how the formula-driven problem of learning expressions from CD-NPI constraints can be reduced to the data-driven ICE model. This reduction allows us to use a host of existing ICE learning algorithms and results in a robust invariant synthesis framework that guarantees to synthesize a provable invariant if one exists.
We present the framwork in Section~\ref{sec:problemstatement} in detail.

However, our CD-NPI learning framework has non-trivial requirements on the verification engine, and building (or adapting) appropriate engines is not straightforward.
To show that our framework is indeed applicable and effective in practice, \textbf{our second contribution is the application of our technique to two real-world verification settings.}

The first setting, presented in Section~\ref{sec:framework}, is the verification of dynamically manipulated data-structures against rich logics that combine properties of structure, separation, arithmetic, and data---an important problem where verification often falls in undecidable theories.
We show how \emph{natural proof verification engines}~\cite{DBLP:conf/pldi/PekQM14}, which are essentially sound-but-incomplete verification engines that translate a powerful undecidable separation logic called \Dryad to decidable logics, can be fit into our framework.
We then implement a prototype of such a natural proof verification engines on top of the program verifier \textsc{Boogie}~\cite{DBLP:conf/fmco/BarnettCDJL05} and demonstrate that this prototype is able to fully automatically verify a large suite of benchmarks, containing standard algorithms for manipulating singly and doubly linked lists, sorted lists, as well as balanced and sorted trees.
Automatically synthesizing invariants for this suite of heap manipulating programs against an expressive separation logic is very challenging, and we do not know of any other current technique that can automatically prove all of them.
Thus, we have to leave a comparison to other approaches for future work.

The second setting is the verification of programs against specifications with universal quantification, which occur, for instance, when defining recursive properties. Again, we implement a prototype over \textsc{Boogie} and demonstrate its effectiveness on a series of benchmarks taken from verification competitions and real-world systems. We describe this application in Section~\ref{sec:quantified_fo} and conclude in Section~\ref{sec:conclusion}.


\subsection*{Related Work}

Techniques for invariant synthesis include abstract interpretation~\cite{DBLP:conf/popl/CousotC77}, interpolation~\cite{mcmillan03}, IC3~\cite{DBLP:conf/vmcai/Bradley11}, predicate abstraction~\cite{DBLP:conf/pldi/BallMMR01}, abductive inference~\cite{DBLP:conf/oopsla/DilligDLM13}, as well as synthesis algorithms that rely on constraint solving~\cite{DBLP:conf/pldi/gulwani08,invgen,DBLP:conf/cav/ColonSS03}. 
Complementing them are data-driven invariant synthesis techniques based on learning, such as Daikon~\cite{DBLP:conf/icse/ErnstCGN00} that learn likely invariants, and \houdini~\cite{DBLP:conf/fm/FlanaganL01} and ICE~\cite{DBLP:conf/cav/0001LMN14} that learn inductive invariants. The latter typically requires a teacher that can generate counter-examples if the conjectured invariant is not adequate or inductive. Classicially, this is possible only when the verification conditions of the program fall in decidable logics.
In this paper, we investigate data-driven invariant synthesis for incomplete verification engines and show that the problem can be reduced to ICE learning if the learning algorithm learns from non-provability information and produces hypotheses in a class that is restricted to positive Boolean formulas over a fixed set of predicates.
Data-driven synthesis of invariants has regained recent interest~\cite{DBLP:conf/cav/SharmaNA12,DBLP:conf/esop/0001GHALN13,DBLP:conf/sas/0001GHAN13,DBLP:conf/cav/0001LMN13,DBLP:conf/cav/0001LMN14,DBLP:conf/cav/0001A14,DBLP:journals/corr/KrishnaPW15,Zhu:2015:LRT:2784731.2784766,DBLP:conf/kbse/PavlinovicLS16,DBLP:conf/pldi/PadhiSM16} and our work addresses an important problem of synthesizing invariants for programs whose verifications conditions fall in undecidable fragments.

Our application to learning invariants for heap programs builds upon \textsc{Dryad}~\cite{DBLP:conf/pldi/Qiu0SM13,DBLP:conf/pldi/PekQM14}, and the natural proof technique line of work for heap verification developed by Qiu~et~al. 
%
Techniques, similar to \textsc{Dryad}, for automated reasoning of dynamically manipulated data structure programs have also been proposed in~\cite{DBLP:conf/pldi/ChuJT15,Chin:2012:AVS:2221987.2222283}. 
However, unlike our current work, none of these works synthesize heap invariants. Given invariant annotations in their respective logics, they provide procedures to validate if the verification conditions are valid.
There has also been a lot of work on synthesizing invariants for separation logic using shape analysis~\cite{Sagiv:1999:PSA:292540.292552,DBLP:journals/jacm/CalcagnoDOY11,DBLP:conf/cav/LeGQC14}. 
However, most of them are tailored for memory safety and shallow properties rather than rich properties that check full functional correctness of data structures. 
Interpolation has also been suggested recently to synthesize invariants involving a combination of data and shape properties~\cite{DBLP:conf/esop/AlbarghouthiBCK15}. 
It is, however, not clear how the technique can be applied to a more complicated heap structure, such as an AVL tree,
where shape and data properties are not cleanly separated but are intricately connected.
Recent work also includes synthesizing heap invariants in the logic from~\cite{DBLP:conf/cav/ItzhakyBINS13} by extending IC3~\cite{DBLP:conf/cav/ItzhakyBRST14,DBLP:conf/cav/KarbyshevBIRS15}. 

In this work, our learning algorithm synthesizes invariants over a fixed set of predicates. When all programs belong to a specific class, such as the class of programs manipulating data structures, these predicates can be uniformly chosen using templates. Investigating automated ways for discovering candidate predicates is a very interesting future direction. Related work in this direction includes recent works~\cite{DBLP:conf/kbse/PavlinovicLS16,DBLP:conf/pldi/PadhiSM16}.



\newcommand\scalemath[2]{\scalebox{#1}{\mbox{\ensuremath{\displaystyle #2}}}}

\section{An Invariant Synthesis Framework for Incomplete Verification Engines}
\label{sec:problemstatement}

In this section, we develop our framework for synthesizing inductive invariants for incomplete verification engines,
using a counter-example guided inductive synthesis approach. 
We do this in the setting where the hypothesis space consists of formulas that are 
Boolean combinations of a fixed set of predicates $\calP$, which need not be finite for the general framework---when developing concrete learning algorithms later, we will assume $\calP$ is a finite set of predicates.
For the rest of this section, let us fix a program $P$ that is annotated with assertions (and possibly with some partial annotations describing pre-conditions, post-conditions, and assertions).
Moreover, we refer to a formula $\alpha$ being weaker (or stronger) than $\beta$ in a logic $\mathcal L$, and by this we mean that $\vdash_\mathcal L \beta \Rightarrow \alpha$ (or $\vdash_\mathcal L \alpha \Rightarrow \beta$), respectively, where $\vdash_\mathcal L \varphi$ means that $\varphi$ is valid in $\mathcal L$.

Figure~\ref{fig:framework} (on Page~\pageref{fig:framework}) 
depicts our general framework of invariant synthesis when verification
is undecidable. We fix several parameters for our verification effort.
First, let us assume a uniform signature for logic, in terms of constant symbols, relation
symbols, functions, and types. We will, for simplicity of exposition,
use the same syntactic logic for the various logics $\calU$, $\calD$, $\calB$
in our framework as well as for the logic $\calH$ used to express invariants.

\begin{figure}[th]
	\centering

	\begin{tikzpicture}
		\node[inner sep=0pt] (left) {\includegraphics[width=0.6\textwidth]{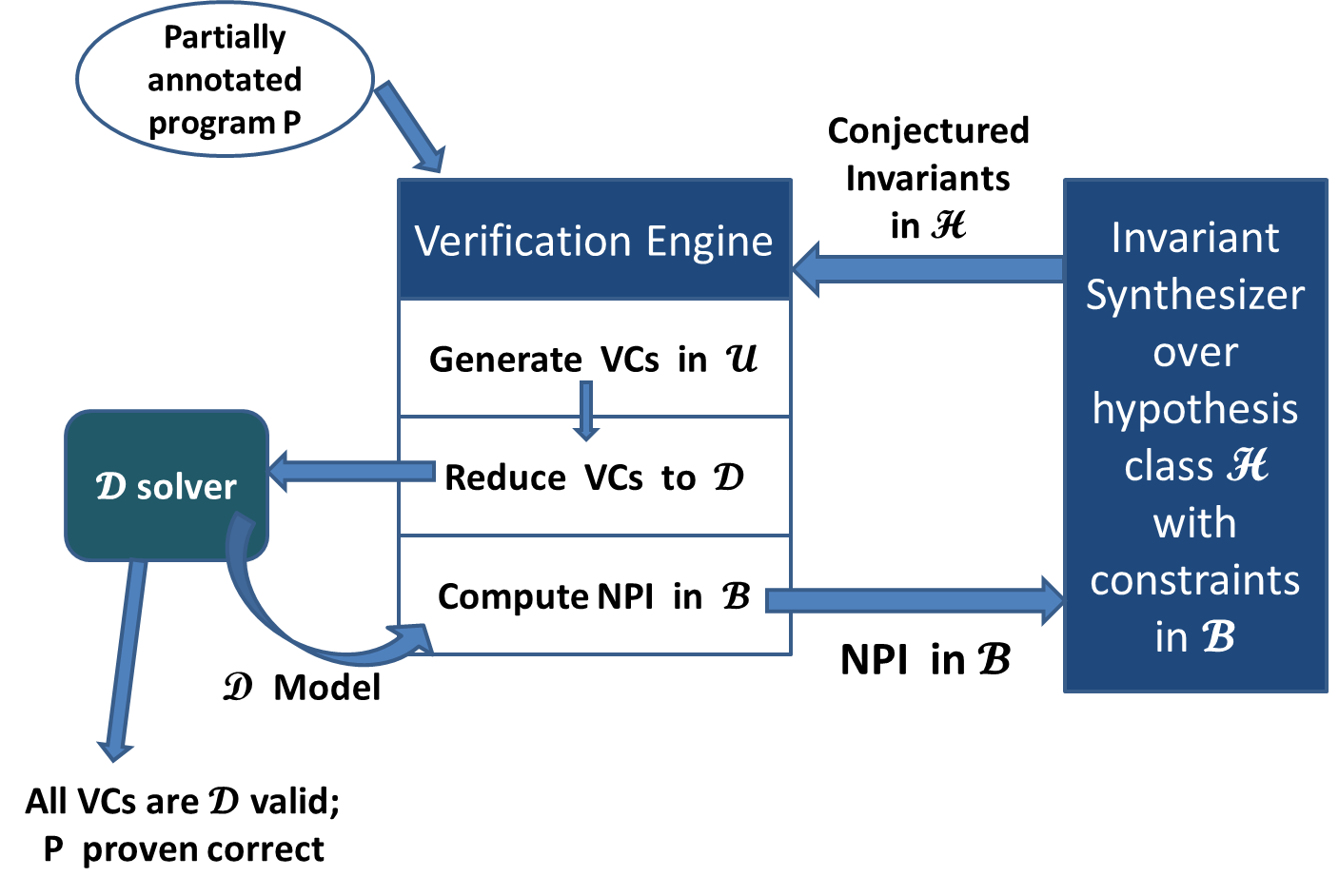}};
		\node[inner sep=0pt, right=3mm] (right) at (left.east) {
			\begin{tabular}[b]{ccp{.3\textwidth}}
				$\calH$ & -- & The hypothesis class of invariants \\
				$\calU$ & -- & The underlying theory of the program; undecidable \\
				$\calD$ & -- & The theory that the verification engine soundly reduces verification conditions to; decidable and can produce models \\
				$\calB$ & -- & The theory of propositional logic that the verification engine uses yo communicate to the invariant synthesis engine
			\end{tabular}
		};
	\end{tikzpicture}

	\caption{A non-provability information (NPI) framework for invariant synthesis when the verification logic is undecidable}  \label{fig:framework}
	
\end{figure}

Let us fix $\calU$ as the underlying theory that is ideally needed for validating the
verification conditions that arise for the program; we presume validity of formulas
in $\calU$ is undecidable.
Since $\calU$ is an undecidable theory, the engine
will resort to sound approximations (e.g., using bounded quantifier instantiations using mechanisms
such as triggers~\cite{DBLP:conf/cade/MouraB07}, bounded
unfolding of recursive functions, or natural proofs~\cite{DBLP:conf/pldi/PekQM14}) to reduce this logical task to a \emph{decidable} theory
$\calD$. This reduction is assumed to be sound in the sense that if the resulting formulas in $\calD$
are valid, then the verification conditions are valid in $\calU$ as well. 
If a formula is found \emph{not valid} in $\calD$, then we require that the logic solver for $\calD$ returns a model
for the negation of the formula.\kern-.06em\footnote{Note that
our framework requires model construction in the theory $\calD$. Hence, incomplete logic
solvers for $\calU$ that simply time out after some time threshold or search for a proof of
a particular kind and give up otherwise are not suitable candidates.}
Note that this model may not be a model for the negation of the formula in $\calU$.

Moreover, we fix a hypothesis class $\calH$ for invariants consisting of \emph{positive} Boolean combination of predicates in a fixed set of
predicates $\calP$.
Note that restricting to \emph{positive} formulas over ${\mathcal P}$ is not a restriction,
as one can always add negations of predicates to ${\mathcal P}$, thus effectively
synthesizing any Boolean combination of predicates. The restriction to positive
Boolean formulas is in fact desirable, as it allows restricting invariants to
\emph{not} negate certain predicates, which is useful when predicates have
intuitionistic definitions (as several recursive definitions of heap properties do).

The invariant synthesis proceeds in rounds, where in each round the synthesizer
proposes invariants in $\mathcal{H}$.
The verification engine generates verification conditions in accordance to these
invariants in the underlying theory $\calU$. It then proceeds to translate them
into the decidable theory $\calD$, and gives them to a solver that decides
validity in the theory $\calD$. If the verification conditions are found to be $\calD$-valid,
then by virtue of the fact that the verification engine reduced VCs in a sound fashion
to $\calD$, we are done proving the program $P$.

However, if the formula is found not to be $\calD$-valid, the solver returns a $\calD$-model 
for the negation of the formula. The verification engine then extracts from this model certain \emph{non-provability information (NPI)}, expressed as Boolean formulas in a Boolean theory $\calB$, that
captures more general reasons why the verification failed 
and eliminates not only the current conjectured invariant but also others that can be inferred to be incorrect from the current verification effort
(the rest of this section is devoted to developing this notion of non-provability information).
This non-provability information is communicated to the synthesizer, which then proceeds to synthesize 
a new conjecture invariant that satisfies the non-provability constraints provided in all previous rounds.
The following example illustrates the logics involved in our framework in the context heap-manipulating programs.

\begin{example}
In a verification setting involving heaps, the logic $\calU$ could  be a rich separation logic with recursive definitions and $\calD$ could be the quantifier-free theory of uninterpreted functions, arithmetic, and sets. The verification engine can reduce verification conditions in $\calU$ to $\calD$ by partially unfolding recursive definitions and expressing heaplets using sets, to obtain sound but incomplete automatic validity checking. 

Futhermore, the theory $\calB$ can be chosen to be the just the propositional theory over a set of predicates $\calP$. The verification engine will then communicate formulas over $\calB$ to the synthesis engine that restrict the class of invariants such that the synthesis engine can generate in the future. \exampleend
\end{example}

In order for the verification engine to extract meaningful non-provability information, we make the following natural assumption, called \emph{normality}, which essentially states that the engine can do at least some minimal Boolean reasoning.

\begin{definition} \label{def:normal_oracle}
A verification engine is \emph{normal} if it satisfies two properties:
\begin{enumerate}[topsep=2pt,itemsep=3pt,partopsep=0ex,parsep=1pt]
	\item \label{itm:normal_precondition}
	if the engine cannot prove the validity of the Hoare triple $\{ \alpha \} s \{ \gamma \}$ and ${}\vdash_\calB \delta \Rightarrow \gamma$, then it cannot prove the validity of the Hoare triple $\{ \alpha \} s \{ \delta \}$; and
	\item \label{itm:normal_postcondition}
	if the engine cannot prove the validity of the Hoare triple $\{ \gamma \} s \{ \beta \}$ and ${}\vdash_\calB \gamma \Rightarrow \delta$, then it cannot prove the validity of the Hoare triple $\{ \delta\} s \{ \beta \}$.
\end{enumerate}
\end{definition}

Intuitively, Condition~\ref{itm:normal_precondition} of Definition~\ref{def:normal_oracle} means that if an oracle cannot prove the validity of $\{ \alpha \} s \{ \gamma \}$, then it cannot prove the validity of any strengthening $\delta$ of the postcondition $\gamma$. Similarly, Condition~\ref{itm:normal_postcondition} means that if an oracle cannot prove the validity of $\{ \gamma \} s \{ \beta \}$, then it cannot prove the validity of any weakening $\delta$ of the precondition $\gamma$.

The remainder of this section is now structured as follows. 
In Section~\ref{sec:CDNPI}, we first develop an appropriate language to communicate non-provability constraints, which allow the learner to appropriately weaken or strengthen a future hypothesis.
It turns out that \emph{pure conjunctions} and \emph{pure disjunctions} over ${\mathcal P}$, which we term \emph{CD-NPI constraints} (conjunctive/disjunctive non-provability information constraints), are sufficient for this purpose. We also describe concretely how the verification engine can extract this non-provability information from $\calD$-models that witness that negations of VCs are satisfiable. Then, in Section~\ref{sec:ice_framework}, we show how to build learners for CD-NPI constraints by reducing this learning problem to another, well-studied learning framework for invariants called ICE learning. We illustrate our framework with an example in Section~\ref{sec:illustrative-example} and finally argue in Section~\ref{sec:framework-correctness} that our framework is sound and guarantees to converge to a provable invariant if one exists.

\subsection{Conjunctive/Disjunctive Non-provability Information}
\label{sec:CDNPI}

We assume that the underlying decidable theory $\calD$ is stronger than propositional theory $\calB$,  meaning that every valid statement in $\calB$ is valid in $\calD$ as well.
The reader may want to keep the following as a running example where $\calD$ is the decidable
theory of uninterpreted functions and linear arithmetic, say. In this setting,
a formula is $\calB$-valid if, when treating atomic formulas as Boolean variables, the formula
is propositionally valid. For instance, $f(x)=y \Rightarrow f(f(x))=f(y)$ will not be $\calB$-valid
though it is $\calD$-valid, while $f(x)=y \vee \neg(f(x)=y)$ is $\calB$-valid.

To formally define CD-NPI constraints and their extraction from a failed verification attempt, let us first introduce the following notation.
For any $\calU$-formula $\varphi$, let $\apprx(\varphi)$ denote the $\calD$-formula that the verification engine generates such
that the $\calD$-validity of $\apprx(\varphi)$ implies the $\calU$-validity of $\varphi$.
Moreover, for any Hoare triple $\{\alpha\} s \{\beta\}$, let $VC(\{\alpha\} s \{\beta\})$ denote the verification
condition corresponding to the Hoare triple that the verification engine generates.

Let us now assume, for the sake of a simpler exposition, that the program has a single annotation hole $A$ where we need to synthesize
an inductive invariant and prove the program correct.
Further, suppose the learner conjectures an annotation $\gamma$ as an inductive invariant for the annotation hole $A$, 
and the verification engine fails to prove the verification condition corresponding to a Hoare triple $\{\alpha\} s \{\beta\}$, 
where either $\alpha$, $\beta$, or both could involve the synthesized annotation. 
This means that the negation of $\apprx( VC( \{ \alpha \} s \{ \gamma \} ) )$ is $\calD$-satisfiable and the verification engine needs to extract non-provability information from a model of it. 
To this end, we assume that every program snippet $s$ has been augmented with a set of ghost variables $g_1, \ldots, g_n$ that track the predicates $p_1, \ldots, p_n$ mentioned in the invariant (i.e., these ghost variables are assigned the values of the predicates).
The valuation $\vec{v}= \langle v_1, \ldots, v_n\rangle$ of the ghost variables in the model before the execution of $s$ and the valuation $\vec{v}' = \langle v'_1, \ldots, v'_n \rangle$ after the execution of $s$ can then be used to derive non-provability information, as we describe shortly.

The type of non-provability information the verification engine extracts depends on where the annotation appears in a Hoare triple $\{\alpha\} s \{\beta\}$. More specifically, the synthesized annotation might appear in $\alpha$, in $\beta$, or in both. We now handle all three cases individually.
\begin{itemize}
	\item 
	Assume the verification of a Hoare triple of the form $\{ \alpha \} s \{ \gamma \}$ fails (i.e., the verification engine cannot prove a verification condition where the pre-condition $\alpha$ is a user-supplied annotation and the post-condition is the synthesized annotation $\gamma$). Then, $\apprx( VC( \{ \alpha \} s \{ \gamma \} ) )$ is not $\calD$-valid, and the decision procedure for $\calD$ would generate a model for its negation. 

	Since $\gamma$ is a positive Boolean combination, the reason why $\vec{v}'$ does not satisfy $\gamma$ is due to the variables mapped to $\textit{false}$ by $\vec{v}'$, as any valuation extending this will not satisfy $\gamma$. Intuitively, this means that the $\calD$-solver is not able to prove the predicates in $P_\textit{false} = \{ p_i \mid v_i' = \textit{false}\}$. In other words, $\{ \alpha \} s \{ \bigvee P_\textit{false} \}$ is unprovable (a witness to this fact is the model of the negation of $\apprx( VC( \{ \alpha \} s \{ \gamma \} ) )$ from which the values $\vec{v}'$ are derived). Note that any invariant $\gamma'$ that is stronger than $\bigvee P_\textit{false}$ will result in an unprovable VC due to the verification engine being normal. Consequently we can choose $\chi =\bigvee P_\textit{false}$ as the weakening constraint, demanding that future invariants should not be stronger than $\chi$.

	The verification engine now communicates $\chi$ to the synthesizer, asking it never to conjecture in future rounds invariants $\gamma''$ that are stronger than $\chi$ (i.e., such that $\not\vdash_{\calB} \gamma'' \Rightarrow \chi$).
	\item
	The next case is when a Hoare triple of the form $\{ \gamma \} s \{ \beta \}$ fails to be proven  (i.e., the verification engine cannot prove a verification condition where the post-condition $\beta$ is a user-supplied annotation and the pre-condition is the synthesized annotation $\gamma$). Using similar arguments as above, the \emph{conjunction} $\eta = \bigwedge \{ p_i \mid v_i = \mathit{true} \} $ of the predicates mapped to $\textit{true}$ by $\vec{v}$ in the corresponding $\calD$-model gives a stronger precondition $\eta$ such that $\{ \eta\} s \{ \alpha \}$ is not provable. Hence, $\eta$ is a valid \emph{strengthening} constraint. The verification engine now communicates $\eta$ to the synthesizer, asking it never to conjecture in future rounds invariants $\gamma''$ that are weaker than $\eta$ (i.e., such that $\not\vdash_{\calB} \eta \Rightarrow \gamma''$).		
	\item
	Finally, consider the case when the Hoare triple is of the form $\{ \gamma \} s \{ \gamma \}$ and fails to be proven (i.e., the verification engine cannot prove a verification condition where the pre- and post-condition is the synthesized annotation $\gamma$). In this case, the verification engine can offer advice on how $\gamma$ can be strengthened \emph{or} weakened to avoid this model. Analogous to the two cases above, the verification engine extracts a pair of formulas $( \eta, \chi )$, called an \emph{inductivity constraint}, based on the variables mapped to $\textit{true}$ by $\vec{v}$ and to $\textit{false}$ by $\vec{v}'$. The meaning of such a constraint is that the invariant synthesizer must conjecture in future rounds invariants $\gamma''$ such that either $\not\vdash_{\calB} \eta \Rightarrow \gamma'' $ or $\not\vdash_{\calB} \gamma'' \Rightarrow \chi $ holds.
\end{itemize}

This leads to the following scheme, where $\gamma$ denotes the conjectured invariant:
\begin{itemize}
	\item When a Hoare triple of the form $\{ \alpha \} s \{ \gamma \}$ fails, the verification engine returns the $\calB$-formula
	\[ \bigvee_{i \mid v_i' = \textit{false}} p_i \]
	as a \emph{weakening constraint}.
	\item When a Hoare triple of the form $\{ \gamma \} s \{ \beta \}$  fails, the verification engine returns the $\calB$-formula
	\[ \bigwedge_{i \mid v_i = \textit{true}} p_i \]
	as a \emph{strengthening constraint}.
	\item When a Hoare triple of the form $\{ \gamma \} s \{ \gamma \}$ fails, the verification engine returns the pair
	\[ (\bigwedge_{i \mid v_i = \textit{true}} p_i,  \bigvee_{i \mid v'_i = \textit{false}} p_i) \]
	of $\calB$-formulas as an inductivity constraint.
\end{itemize}

It is not hard to verify that the above  formulas are proper strengthening and weakening constraints,
in the sense that \emph{any} inductive invariant must satisfy these constraints.
This motivates the following form of non-provability information.

\begin{definition}[CD-NPI Samples]
Let $\mathcal P$ be a set of predicates. A \emph{CD-NPI sample} (short for \emph{conjunction-disjunction-NPI sample}) is a triple $\mathfrak S = (W, S, I)$  consisting of 
\begin{itemize}
	\item a finite set $W$ of disjunctions over $\mathcal P$ (weakening constraints);
	\item a finite set $S$ of conjunctions over $\mathcal P$ (strengthening constraints); and
	\item a finite set $I$ of pairs, where the first element is a conjunction and the second is a disjunction over $\mathcal P$ (inductivity constraints).
\end{itemize}

An annotation $\gamma$ is \emph{consistent} with a CD-NPI sample $\mathfrak S = (W, S, I)$ if
\begin{itemize}
 \item $\not \vdash_\calB \gamma \Rightarrow \chi$ for each $\chi \in W$;
 \item $\not\vdash_\calB \eta \Rightarrow \gamma$ for each $\eta \in S$; and
 \item $\not\vdash_\calB \eta \Rightarrow \gamma$ or $\not\vdash_\calB \gamma \Rightarrow \chi$ for each $(\eta, \chi) \in I$.
\end{itemize}
\end{definition}

A CD-NPI learner is an effective procedure that synthesizes, given an CD-NPI sample, an annotation $\gamma$ consistent with the sample.
In our framework, the process of proposing candidate annotations and checking them
repeats until the learner proposes a valid annotation or it detects that no valid annotation exists (e.g., if the class of candidate annotations is finite and all annotations are exhausted). 
We comment on using an CD-NPI learner in this iterative fashion below.

\subsection{Building CD-NPI Learners}
\label{sec:ice_framework}
Let us now turn to the problem of building efficient learning algorithms for CD-NPI constraints. To this end, we assume  that the set of predicates ${\mathcal P}$ is finite.

Roughly speaking, the CD-NPI learning problem is to synthesize annotations that are positive Boolean combinations of predicates
in ${\mathcal P}$ and that are consistent with given CD-NPI samples.
Though this is a learning problem where samples are \emph{formulas}, in this section we will reduce CD-NPI
learning to a learning problem from \emph{data}. In particular, we will show that CD-NPI learning
reduces to the ICE learning framework for learning positive Boolean formulas.
The latter is a well-studied framework, and the reduction allows us to use efficient learning algorithms
developed for ICE learning in order to build CD-NPI learners.

We now first recap the ICE-learning framework and then reduce CD-NPI learning to ICE learning. Finally, we briefly sketch how the popular \houdini algorithm can be seen as an ICE learning algorithm, which, in turn, allows us to \houdini as an CD-NPI learning algorithm.

\subsubsection*{The ICE learning framework}
Although the ICE learning framework~\cite{DBLP:conf/cav/0001LMN14} is a general framework for learning inductive invariants, we consider here the case of learning Boolean formulas.
To this end, let us fix a set $B$ of Boolean variables, and let $\calH$ be a subclass of positive Boolean formulas over $B$, called the hypothesis class, which specifies the admissible solutions to the learning task.

The objective of the (passive) ICE learning algorithm is to learn a formula in $\calH$ from a sample of positive examples, negative examples, and implication examples. More formally, if $\mathcal V$ is the set of valuations $v \colon B \to \{\mathit{true}, \mathit{false}\}$ (mapping variables in $B$ to true or false),
then an \emph{ICE sample} is a triple $\mathcal S = (S_+, S_-, S_\Rightarrow)$ where
\begin{itemize}
	\item $S_+ \subseteq \mathcal V$ is a set of \emph{positive examples};
	\item $S_- \subseteq \mathcal V$ is a set of \emph{negative examples}; and
	\item $S_\Rightarrow \subseteq \mathcal V \times \mathcal V$ is a set of \emph{implications}.
\end{itemize}
Note that positive and negative examples are \emph{concrete} valuations of the variables $B$, and the implication examples are pairs of such concrete valuations.

A formula $\varphi$ is said to be \emph{consistent with an ICE sample $\mathcal S$} if it satisfies the following three conditions:\footnote{In the following, $\models$ denotes the usual satisfaction relation.}
\begin{itemize}
	\item $v \models \varphi$ for each $v \in S_+$;
	\item $v \not \models \varphi$ for each $v \in S_-$; and
	\item $v_1 \models \varphi$ implies $v_2 \models \varphi$ for each $(v_1, v_2) \in S_\Rightarrow$.
\end{itemize}

In algorithmic learning theory, one distinguished between \emph{passive learning} and \emph{iterative learning}. The former refers to a learning setting in which a learning algorithm is confronted with a finite set of data and has to learn a concept that is consistent with this data.
Using our terminology, the \emph{passive ICE learning problem} for a hypothesis class $\calH$ is then
\begin{quote}
	 \textit{``given an ICE sample $\mathcal S$, find a formula in $\calH$ that is consistent with $\mathcal S$''}.
\end{quote}
Recall that we here require the learner to learn positive Boolean formulas, which is slightly stricter than the original definition~\cite{DBLP:conf/cav/0001LMN14}.

Iterative learning, on the other hand, is the iteration of passive learning where new data is added to the sample from one iteration to the next. In a verification context, this new data is generated by the verification engine in response to incorrect annotations and used to guide the learning algorithm towards an annotation that is adequate to prove the program. To reduce our learning framework to ICE learning, it is therefore sufficient to reduce the (passive) CD-NPI learning problem described above to the passive ICE learning problem. We do this next.

\subsubsection*{Reduction of passive CD-NPI learning to passive ICE learning}
\label{sec:CD-NPI2ICE}

Let $\calH$ be a subclass of positive Boolean formulas.
We reduce the CD-NPI learning problem for $\calH$ to the ICE learning problem for $\calH$.
The main idea is to
\begin{enumerate*}[label={(\alph*)}]
	\item treat each predicate $p \in \calP$ as a Boolean variable for the purpose of ICE learning and
	\item to translate a CD-NPI sample $\mathfrak G$ into an \emph{equi-consistent} ICE sample $\mathcal S_\mathfrak S$, meaning that a positive Boolean formula is consistent with $\mathfrak S$ if and only if it is consistent with $\mathcal S_\mathfrak S$.
\end{enumerate*}
Then, learning a consistent formula in the CD-NPI framework for the hypothesis class $\calH$ 
reduces to learning consistent formulas in $\calH$ in the ICE learning framework.

The following lemma will help translate between the two frameworks.  Its proof is straightforward,
and follows from the fact that for any \emph{positive} formula $\alpha$, if a valuation $v$ sets a larger
subset of propositions to true than $v'$ does and $v' \models \alpha$, then $v \models \alpha$ as well.

\begin{lemma}\label{lem:lemma_1}
Let $v$ be a valuation of $\mathcal P$ and $\alpha$ be a positive Boolean formula over $\mathcal P$.
Then, the following holds:
\begin{itemize}[topsep=2pt,itemsep=3pt,partopsep=0ex,parsep=1pt]
 \item $v \models \alpha$ if and only if ${}\vdash_\calB (\bigwedge_{p \mid v(p)=\mathit{true}} p ) \Rightarrow \alpha$ (and, thus, $v \not \models \alpha$ if and only if \hbox{${}\not\vdash_\calB (\bigwedge_{p \mid v(p)=\mathit{true}} p ) \Rightarrow \alpha)$}.
 \item $v \models \alpha$ if and only if ${}\not \vdash_\calB \alpha \Rightarrow (\bigvee_{p \mid v(p)=\mathit{false}} p )$.
\end{itemize}
\end{lemma}

This motivates our translation, which relies on two functions, $d$ and $c$.
The function $d$ translates a disjunction $\bigvee J$, where $J \subseteq \mathcal P$ is a subset of propositions, into the valuation
\[ d \bigl( \bigvee J \bigr) = v \text{ with $v(p) = \textit{false}$ if and only if $p \in J$}. \]
The function $c$ translates a conjunction $\bigwedge J$, where $J \subseteq \mathcal P$, into the valuation
\[ c \bigl( \bigwedge J \bigr) = v \text{ with $v(p) = \textit{true}$ if and only if $p \in J$}. \]
By substituting $v$ in Lemma~\ref{lem:lemma_1} with $c(\bigwedge J)$ and $d(\bigvee J)$, respectively, one immediately obtains the following result.

\begin{lemma} \label{lem:pos-boolean}
Let $J \subseteq \mathcal P$ and $\alpha$ be a positive Boolean formula over $\mathcal P$.
Then, the following holds:
\begin{enumerate*}[label={(\alph*)}]
 \item $c \bigl( \bigwedge J \bigr) \not \models \alpha$ if and only if ${}\not\vdash_\calB \bigwedge J \Rightarrow \alpha$, and 
 \item $d \bigl( \bigvee J \bigr) \models \alpha$ if and only if ${}\not \vdash_\calB \alpha \Rightarrow \bigvee J$.
\end{enumerate*}
\end{lemma}

Based on the functions $c$ and $d$, the translation of a CD-NPI sample into an equi-consistent ICE sample is as follows.

\begin{definition} \label{def:ICE_sample}
Given a CD-NPI sample $\mathfrak S = (W, S, I)$, the ICE sample $\mathcal S_\mathfrak S = (S_+, S_-, S_\Rightarrow)$ is defined by
$S_+ = \bigl \{ d(\bigvee J) \mid \bigvee J \in W \bigr \}$,
$S_- = \bigl \{ c(\bigwedge J) \mid \bigwedge J \in S \bigr \}$, and
$S_\Rightarrow = \bigl \{ \bigl( c(\bigwedge J_1), d(\bigvee J_2) \bigr ) \mid (\bigwedge J_1, \bigvee J_2) \in I \bigr \}$.
\end{definition}

By virtue of the lemma above, we can now establish the correctness of the reduction from the CD-NPI learning problem to the ICE learning problem.

\begin{theorem} \label{thm:ICE_correct}
Let $\mathfrak S = (W, S, I)$ be a CD-NPI sample, $\mathcal S_\mathfrak S = (S_+, S_-, S_\Rightarrow)$ the ICE sample as in Definition~\ref{def:ICE_sample}, $\gamma$ a positive Boolean formula over $\mathcal P$. Then, $\gamma$ is consistent with $\mathfrak S$ if and only if $\gamma$ is consistent with $\mathcal S_\mathfrak S$.
\end{theorem}

\begin{proof} 
Let $\mathfrak S = (W, S, I)$ be an CD-NPI sample, and let $\mathcal S_\mathfrak S = (S_+, S_-, S_\Rightarrow)$ the ICE sample as in Definition~\ref{def:ICE_sample}. Moreover, let $\gamma$ be a positive Boolean formula.
We prove Theorem~\ref{thm:ICE_correct} by considering each weakening, strengthening, and inductivity constraint together with their corresponding positive, negative, and implication examples individually.

\begin{itemize}[itemsep=.5\baselineskip]
	\item Pick a weakening constraint $\bigvee J \in W$, and let $v \in S_+$ with $v = d(\bigvee J)$ be the corresponding positive sample. Moreover, assume that $\gamma$ is consistent with $\mathfrak S$ and, thus, $\not \vdash_\calB \gamma \Rightarrow  \bigvee J$. By Lemma~\ref{lem:pos-boolean}, this is true if and only if $d \bigl( \bigvee J \bigr) \models \gamma$. Hence, $v \models \gamma$.
	
	Conversely, assume that $\gamma$ is consistent with $\mathcal S$. Thus, $v \models \gamma$, which means $d \bigl( \bigvee J \bigr) \models \gamma$. By Lemma~\ref{lem:pos-boolean}, this is true if and only if $\not \vdash_\calB \gamma \Rightarrow  \bigvee J$.
	\item
	Pick a strengthening constraint $\bigwedge J \in S$, and let $v \in S_-$ with $v = c(\bigwedge J)$ be the corresponding negative sample. Moreover, assume that $\gamma$ is consistent with $\mathfrak S$ and, thus, $\not\vdash_\calB \bigwedge J \Rightarrow \gamma$. By Lemma~\ref{lem:pos-boolean}, this is true if and only if $c \bigl( \bigwedge J \bigr) \not \models \gamma$. Hence, $v \not\models \gamma$.
	
	Conversely, assume that $\gamma$ is consistent with $\mathcal S$. Thus, $v \not\models \gamma$, which means $c \bigl( \bigwedge J \bigr) \not \models \gamma$. By Lemma~\ref{lem:pos-boolean}, this is true if and only if $\not\vdash_\calB \bigwedge J \Rightarrow \gamma$.
	\item
	Following the definition of implication, we split the proof into two cases, depending on whether $\not\vdash_\calB \bigwedge J \Rightarrow \gamma$ or $\not\vdash_\calB \gamma \Rightarrow \bigvee J$ (and $v_1 \not\models \gamma$ or $v_2 \models \gamma$ for the reserve direction). However, the proof in the former case is the same as the proof for strengthening constraints, while the proof of latter case is the same as the proof for weakening. Hence, combining both proofs immediately yields the claim. \qed
\end{itemize}
\end{proof}

\subsubsection*{ICE learners for Boolean formulas}
The reduction above allows us to use any ICE learning algorithm in the literature that
synthesizes positive Boolean formulas. As we mentioned earlier, we can add the negations
of predicates as first-class predicates, and hence synthesize invariants over the
class of all Boolean combinations of a finite set of predicates as well.

The problem of passive ICE learning for one round, synthesizing a formula that satisfies
the ICE sample, can usually be achieved efficiently and in a variety of ways. However,
the crucial aspect is not the complexity of learning in one round, but the \emph{number}
of rounds it takes to converge to an adequate invariant that proves the program correct.
When the set $\calP$ of candidate predicates is large (hundreds in our experiments),
since the number of Boolean formulas over $\mathcal P$ is doubly exponential in $n = |\calP|$, building
an effective learner is not easy.
However, there is one class of formulas that are particularly amenable to efficient
ICE learning---learning \emph{conjunctions of predicates over $\mathcal P$}. In this case,
there are ICE learning algorithms that promise learning the invariant (provided one
exists expressible as a conjunct over $\mathcal P$) in $n+1$ rounds.  Note that this learning
is essentially finding an invariant in a hypothesis class $\calH$ of size $2^n$ in $n+1$
rounds.

\houdini~\cite{DBLP:conf/fm/FlanaganL01} is such a learning algorithm for conjunctive formulas. Though it is 
typically seen as a particular way to synthesize invariants, it is a prime example of an ICE learner for conjuncts, as described in the work by Garg~et~al.~\cite{DBLP:conf/cav/0001LMN14}.
In fact, Houdini is similar to the classical PAC learning algorithm for 
conjunctions~\cite{Kearns:1994:ICL:200548}, but extended to the ICE model by handling implication counterexamples.
More precisely, given an ICE sample $\mathcal S = (S_+, S_-, S_\Rightarrow)$, \houdini computes the largest conjunctive formula $\varphi$ in terms of the number of Boolean variables occurring in $\varphi$ (i.e., the semantically strongest conjunctive formula) that is consistent with $\mathcal S$ in the following way. First, it computes the largest conjunction $\varphi$ that is consistent with the positive examples (i.e., $v \models \varphi$ for all $v \in S_+$); note that this conjunction is unique. Next, \houdini checks whether the implications are satisfied. If this is not the case, then we know for each non-satisfied implication $(v_1, v_2) \in S_\Rightarrow$ that $v_2$ has to be classified positively because $v_1$ belongs to every set that includes $S_+$. Hence, \houdini adds all such $v_2$ to $S_+$, resulting in a new set $S'_+$. Subsequently, it constructs the largest conjunction $\varphi'$ that is consistent with the positive examples in $S'_+$ (i.e., $v \models \varphi'$ for all $v \in S'_+$). \houdini repeats this procedure until it arrives at the largest conjunctive formula $\varphi^\ast$ that is consistent with $S_+$ and $S_\Rightarrow$ (again, note that this set is unique). Finally, \houdini checks whether each negative example violates $\varphi^\ast$  (i.e., $v \not\models \varphi^\ast$ for all $v \in S_-$). If this is the case, $\varphi^\ast$ is the largest conjunctive formula over $\mathcal B$ that is consistent with $\mathcal S$; otherwise, no consistent conjunctive formula exists.
The time \houdini spends in each round is \emph{polynomial} and, furthermore, when used in an iterative setting, is guaranteed to converge in at most $n+1$ rounds or report that no conjunctive invariant over $\mathcal P$ exists.
We use this ICE learner to build a CD-NPI learner for conjunctions.

\subsection{An Illustrative Example} \label{sec:illustrative-example}

\begin{figure}[t]
\begin{lstlisting}[mathescape=true,basicstyle=\small]
int A[], B[];    
int N;    axiom ($N > 0$);
bool inImage(int i) { return true; }

procedure inverse()
requires ($\forall x, y. ~0 \leq x < y < N 
                           \implies A[x] \not= A[y]$);  // A is injective 
requires ($\forall x. ~0 \leq x < N \wedge inImage(x) \implies (\exists y. ~0 \leq y < N \wedge A[y] = x)$); // A is surjective 
ensures ($\forall x, y. ~0 \leq x < y < N \implies B[x] \not= B[y]$);  // B is injective
{
  int i = 0;
  while (i < N)
  SynthesizeInv($\forall x. ~0 \leq x < i \implies B[A[x]] = x$); // $b_1$
  {
  	B[A[i]] = i;    
  	i = i + 1;
  }
  SynthesizeInv($\forall x. ~0 \leq x < N \implies A[B[x]] = x$, // $b_2$
                  $\forall x. ~0 \leq x < N \wedge inImage(x) \implies A[B[x]] = x$); // $b_3$
  return;
}

\end{lstlisting}

\caption{Synthesizing invariants for the program that constructs an inverse B of an injective, surjective function A~\cite{vscomp2010}.} \label{fig:example}
\end{figure}

Figure~\ref{fig:example} illustrates an example program of the verified software competition~\cite{vscomp2010} that given an injective, surjective function A returns the inverse B of the function A. The post-condition of this program expresses that the function B is injective. To prove this program correct, one needs to specify adequate invariants at the loop header and before the return statement in the function \textit{inverse} in the program. We wish to synthesize these invariants.
For simplicity, let us assume we are provided a small set of predicates as building blocks of the invariants to synthesize---$b_1$ for the loop invariant and $b_2$, $b_3$ for the invariant at the return statement. Our task, therefore, is to synthesize adequate invariants for this program over these predicates.\kern-.06em\footnote{In general, one starts with a much larger set of candidate predicates that are automatically generated using program/specification-dependent heuristics.}

Clearly, the verification conditions of this program are undecidable. In fact, the constant Boolean function \textit{inImage} is crucially required to validate certain verification conditions in \textsc{Boogie} because it triggers appropriate quantifier instantiations in the surjectivity condition.

Now, suppose the learner conjectures the loop invariant $\gamma_L = b_1$ and the invariant at the return statement $\gamma_R = b_2 \wedge b_3$. Moreover, suppose that the verification condition along the path from the loop exit to the return statement, though valid in the undecidable theory $\calU$ (cf.\ Figure~\ref{fig:framework}), is not provable in the decidable theory $\calD$ (one that has instantiated quantifiers with ground terms). 
The $\calD$-solver returns a model for the negation of the verification condition that captures this non-provability information. The verification engine gleans this model---it looks for the values assigned to the predicate variables in the model, and from this information constructs, in general, a CD-NPI constraint for the learner to learn from. For this particular verification, the verification engine extracts a pair of formulas $( \eta, \chi )$ where $\eta = b_1$ and $\chi = b_2$, and communicates this as an inductivity constraint to the learner. Intuitively, this  constraint means that the verification condition obtained by substituting $\gamma_L$ with $\eta$ and $\gamma_R$ with $\chi$ is itself not provable. Hence, in subsequent rounds, the learner needs to conjecture only such invariants where $\gamma_L$ is not weaker than $\eta$ (i.e., $\not \vdash_\calB b_1 \Rightarrow \gamma_L$) or $\gamma_R$ is not stronger than $\chi$ (i.e., $\not \vdash_\calB \gamma_R \Rightarrow b_2$). 

The learner works by reducing the CD-NPI passive learning problem to ICE learning over a sample over the given set of predicates. Concretely, the inductivity constraint $(b_1, b_2)$ is reduced to an implication constraint $((1,0,0), (1,0,1))$ in the ICE setting, where each datapoint in the ICE sample has values for the predicates $b_1$, $b_2$, and $b_3$, respectively.
In the next round, let us assume the learner conjectures the invariants $\gamma_L = b_1$ and $\gamma_R = b_3$. Note these conjectures satisfy both the ICE constraints and the CD-NPI constraints. In this case, it turns out that the verification conditions along all program paths using these invariants can be proved valid in the theory $\calD$. As a result, our invariant synthesis procedure terminates with $\gamma_L$ and $\gamma_R$ as adequate inductive invariants.

\subsection{Main Result}
\label{sec:framework-correctness}
To state the main result of this paper, let us assume that the set $\calP$ of predicates is finite. We comment on the case of infinitely many predicates below.

\begin{theorem} \label{thm:CD-NPI-main-result}
Assume a normal verification engine for a program $P$ to be given. Moreover, let $\calP$ be a finite set of predicates over the variables in $P$ and $\calH$ a hypothesis class consisting of positive Boolean combinations of predicates in $\calP$. 
If there exists an annotation in $\calH$ that the verification engine can use to prove $P$ correct, then the CD-NPI framework described in Section~\ref{sec:CDNPI} is guaranteed to converge to such an annotation in finite time.
\end{theorem}

\begin{proof}[Proof of Theorem~\ref{thm:CD-NPI-main-result}]
The proof proceeds in two steps. First, we show that a normal verification engine is \emph{honest}, meaning that the non-provability information returned by such an engine does not rule out any adequate and provable annotation. Second, we show that any consistent learner (i.e., a learner that only produces consistent hypotheses), when paired with an honest verification engine, makes \emph{progress} from one round to another. Finally, we combine both results to show that the framework eventually converges to an adequate and provable annotation.

\paragraph{Honesty of the verification engine}
We show honesty of the verification engine individually for each type of constraint by contradiction.
\begin{itemize}
	\item Suppose that the verification replies to a candidate invariant $\gamma$ proposed by the learner with a weakening constraint $\chi$ because it could not prove the validity of the Hoare triple $\{ \alpha \} s \{ \gamma \}$. This effectively forces any future conjecture $\gamma'$ to satisfy $\not\vdash_\calB \gamma' \Rightarrow \chi$.
	
	Now, suppose that there exists an invariant $\delta$ such that $\vdash_\calB \delta \Rightarrow \chi$ and the verification engine can prove the validity of $\{ \alpha \} s \{ \delta \}$ (in other words, the adequate invariant $\delta$ is ruled out by the weakening constraint $\chi$). Due to the fact that the verification engine is normal (in particular, by contraposition of Part~\ref{itm:normal_precondition} of Definition~\ref{def:normal_oracle}), this implies that the verification engine can also prove the validity of $\{ \alpha \} s \{ \chi \}$. However, this is a contradiction to $\chi$ being a weakening constraint.
	\item Suppose that the verification engine replies to a candidate invariant $\gamma$ proposed by the learner with a strengthening constraint $\eta$ because it could not prove the validity of the Hoare triple $\{ \gamma \} s \{ \beta \}$. This effectively forces any future conjecture $\gamma$ to satisfy $\not\vdash_\calB \eta \Rightarrow \gamma'$.
	
	Now, suppose that there exists an invariant $\delta$ such that $\vdash_\calB \eta \Rightarrow \delta$ and the verification engine can prove the validity of $\{ \delta \} s \{ \beta \}$ (in other words, the adequate invariant $\delta$ is ruled out by the weakening constraint $\eta$). Due to the fact that the verification engine is normal (in particular, by contraposition of Part~\ref{itm:normal_postcondition} of Definition~\ref{def:normal_oracle}), this implies that the verification engine can also prove the validity of $\{ \eta \} s \{ \beta \}$. However, this is a contradiction to $\eta$ being a strengthening constraint.
	\item Combining the arguments for weakening and strengthening constraints immediately results in a contradiction for the case of inductivity constraints as well.
\end{itemize}

\paragraph{Progress of the learner}
Now suppose that the learning algorithm is consistent, meaning that it always produces an annotation that is consistent with the current sample. Moreover, assume that the sample in iteration $i \in \mathbb N$ is $\mathfrak S_i$ and the learner produces the annotation $\gamma_i$. If $\gamma_i$ is inadequate to prove the program correct, the verification engine returns a constraint $c$. The learner adds this constraint to the sample, obtaining the sample $\mathfrak S_{i+1}$ of the next iteration.

Since verification with $\gamma_i$ failed, which is witnessed by $c$, we know that $\gamma_i$ is not consistent with $c$. The next conjecture $\gamma_{i+1}$, however, is guaranteed to be consistent with $\mathfrak S_{i+1}$ (which contains $c$) because the learner is consistent. Hence, $\gamma_i$ and $\gamma_{i+1}$ are semantically different. Using this argument repeatedly shows that each annotation $\gamma_i$ that a consistent learner has produced is semantically different from any previous annotation $\gamma_j$ for $j < i$.

\paragraph{Convergence}
We first make two observations.
\begin{enumerate}
	\item \label{itm:proof:main-result:1} The number of semantically different hypotheses in the hypothesis space $\calH$ is finite because the set $\calP$ is finite. Recall that $\calH$ is the class of all positive Boolean combinations of predicates in $\calP$.
	\item \label{itm:proof:main-result:2} Due to the honesty of the verification engine, every annotation that the verification engine can use to prove the program correct is guaranteed to be consistent with any sample produced during the learning process.
\end{enumerate}

Now, suppose that there exists an annotation that the verification engine can use to prove the program correct.
Since the learner is consistent, all conjectures produced during the learning process are semantically different. Thus, the learner will at some point have exhausted all incorrect annotations in $\calH$ (due to Observation~\ref{itm:proof:main-result:1}). By assumption, however, there exists at least one annotation that the verification engine can use to prove the program correct. Moreover, any such annotation is guaranteed to be consistent with the current sample (due to Observation~\ref{itm:proof:main-result:2}). Thus, the annotation conjectured next is necessarily one that the verification engine can use to prove the program correct. \qed
\end{proof}

Under certain, realistic assumptions on the CD-NPI learning algorithm, Theorem~\ref{thm:CD-NPI-main-result} remains true even if the number of predicates is infinite. An example of such an assumption is that the learning algorithm always conjectures a smallest consistent annotation with respect to some fixed total order on $\calH$. In this case, one can show that such a learner will at some point have proposed all inadequate annotation up to the smallest annotation the verification engine can use to prove the program correct. It will then conjecture this annotation in the next iteration.
We refer the reader to Löding, Madhusudan, and Neider~\cite{DBLP:conf/tacas/LodingMN16} for details on further strategies that ensure convergence.


\section{Application: Learning Invariants that Aid Natural Proofs for Heap Reasoning}
\label{sec:framework}

We now develop an implementation of our learning framework for verification
engines based on natural proofs for heap reasoning~\cite{DBLP:conf/pldi/Qiu0SM13,DBLP:conf/pldi/PekQM14}.
We first provide some background on the separation logic \Dryad and natural proofs, which is a sound but incomplete verification procedure.
Then, we describe how to implement our verification framework using a natrual proofs verification engine. In particular, we describe how to automatically generate
suitable predicates for these programs, which serve as the building blocks of the invariants we seek to synthesize.
Finally, we present an empirical evaluation of our implementation on an extensive set of standard algorithms on dynamic data structures, such as searching, inserting, or deleting items in lists and trees.

\subsubsection*{Background: Natural Proofs and Dryad}
\label{sec:dryad}

\Dryad~\cite{DBLP:conf/pldi/Qiu0SM13,DBLP:conf/pldi/PekQM14} is a dialect of separation logic 
that comes with a heaplet semantics and allows expressing second order properties such as pointing to a list (or list segment) using recursive functions and predicates. 
The syntax of \Dryad is a standard separation logic syntax with a few restrictions, such as disallowing negations inside recursive definitions and in sub-formulas connected by spatial conjunction (see~\cite{DBLP:conf/pldi/PekQM14} for more details about the \Dryad syntax).  
\Dryad is expressive enough to state a variety of data-structures (singly and doubly linked lists, sorted lists, binary search trees,
AVL trees, maxheaps, treaps, etc.), recursive definitions over them that map to numbers (length, height, etc.), as well as data stored within
the heap (the multiset of keys stored in lists, trees, etc.).

Natural proofs~\cite{DBLP:conf/pldi/Qiu0SM13,DBLP:conf/pldi/PekQM14} is a sound but incomplete strategy for deciding satisfiability of \Dryad formulas. The first step of the natural proof verifier is to convert all predicates and functions in a \textsc{Dryad}-annotated program to \emph{classical logic}. This translation introduces \emph{heaplets} (modeled as sets
of locations) explicitly in the logic. Furthermore, it introduces assertions that demand that the
accesses of each method are contained in the heaplet implicitly defined by its precondition (taking into
account newly allocated or freed nodes), and that at the end of the program, the modified heaplet precisely 
matches the implicit heaplet defined by the post-condition. 

The second step of the natural proof verifier is to perform \emph{transformations}
on the program and translate it to \textsc{Boogie}~\cite{DBLP:conf/fm/FlanaganL01}, an intermediate verification language that handles
proof obligations using automatic theorem provers (typically SMT solvers).
This transformation essentially performs three tasks:
\begin{enumerate*}[label={(\alph*)}]
	\item it abstracts all recursive definitions on the heap using uninterpreted functions and introduces finite-depth unfoldings of recursive definitions at every place in the code where locations are dereferenced,
	\item it models heaplets and other sets using a decidable theory of maps, and
	\item it inserts \emph{frame reasoning} explicitly in the code that allows the verifier to derive that certain properties continue to hold across a heap update (or function call) using the heaplet that is modified.
\end{enumerate*}
The resulting \textsc{Boogie} program is a program with no recursive definitions, where all verification
conditions are in decidable logics, and where the logic engine can return models when formulas are satisfiable. 
The program can be verified if supplied with correct inductive loop-invariants and
adequate pre/post conditions. 

The described procedure has been implemented in a fully automatic tool, called \textsc{VCDryad}. \textsc{VCDryad} extends \textsc{VCC}~\cite{DBLP:conf/tphol/CohenDHLMSST09} and converts C programs annotated in \textsc{Dryad} to \textsc{Boogie} programs via the natural proof transformations described above.
It is important to note, however, that \textsc{VCC} introduces some quantification to define the memory model and semantics of C, but this does not typically derail decidable reasoning.
We refer the reader to~\cite{DBLP:conf/pldi/Qiu0SM13,DBLP:conf/pldi/PekQM14} for more details.

\subsubsection*{Learning Heap Invariants}
We have implemented a prototype of our CD-NPI framework over \textsc{VCDryad} and the \textsc{Boogie} program verifier. This prototype takes a C program annotated in \textsc{Dryad} as input and uses \textsc{VCDryad} to convert it to a \textsc{Boogie} program. Then, it applies our transformation to the ICE learning framework and automatically generates a set $\calP$ of \emph{predicates} (as described shortly), which serve as the basic building blocks of our invariants.
Finally, it pairs the \textsc{Boogie} verifier with an invariant synthesis engine, \textsc{Houdini} in our case, to learn an inductive invariant. Note that after the \textsc{VCDryad}-transformation, \textsc{Boogie} satisfies the requirements on verification engines of our framework.

The set $\calP$ of predicates is generated from generic templates, shown in Figure~\ref{fig:predicates}, which are instantiated using all combinations of program variables that occur in the program being verified.
The templates define a fairly exhaustive set of predicates, including
\begin{itemize}
	\item properties of the store (equality of pointer variables, equality and inequalities between integer variables, etc.),
	\item shape properties (singly and doubly linked lists and list segments, sorted lists, trees, BST, AVL, treaps, etc.),
	\item and recursive definitions that map data structures to numbers (keys/data stored in a structure, lengths of lists and list segments, height of trees) involving arithmetic relationships and set relationships.
\end{itemize}
In addition, there are also predicates describing heaplets of various structures (with suffix \textit{\_heaplet}), involving set operations, disjointness, and equalities. The structures and predicates are extensible, of course, to any recursive definition expressed in \textsc{Dryad}.

\newcommand{\pf}{\textit{pf}}
\newcommand{\df}{\textit{df}}
\begin{figure}[ht!] 
\def\arraystretch{1.3}

\begin{displaymath}
\begin{array}{l}

\begin{array}{c}
 x,y \in \textit{PointerVars}  \hspace{0.5cm}  \vec{x},\vec{y},\vec{z} \in \textit{PointerVars}^*  \hspace{0.5cm} \pf \in \textit{PointerFields} \hspace{0.5cm} \textit{key},~ \df \in \textit{DataFields}  \\[-0.2em]
i,j \in \textit{IntegerVars} ~\cup~ \{ 0, \textsf{IntMax}, \textsf{IntMin} \}
\end{array}

\\[0.6cm]

\begin{array}{lrll}
& \textit{listshape}(\vec{x}) & := & \textsf{LinkedList}(x_1)  ~~\mid~~  \textsf{DoublyLinkedList}(x_1) ~~\mid~~  \textsf{SortedLinkedList}(x_1) \\[-0.25em] &&& 
~~\mid~~  \textsf{LinkedListSeg}(x_1,x_2) ~~\mid~~ \textsf{DoublyLinkedListSeg}(x_1,x_2)   \\[-0.25em] &&& 
    ~~\mid~~ \textsf{SortedLinkedListSeg}(x_1,x_2)
\\[0.05cm]
&\textit{treeshape}(x) & := & \textsf{BST}(x) ~~\mid~~ \textsf{AVLtree}(x) ~~\mid~~ \textsf{Treap}(x) 
\\[0.05cm]
&\textit{shape}(\vec{x}) & := & \textit{listshape}(\vec{x}) ~~\mid~~ \textit{treeshape}(\vec{x}) 
\\[0.05cm]
&\textit{size}(\vec{x}) & := & \textit{listshape}\_\textsf{length}(\vec{x}) ~~\mid~~ \textit{treeshape}\_\textsf{height}(\vec{x})
\end{array}
\end{array}
\end{displaymath}

\begin{displaymath}
\hspace{1em}
\begin{array}{l}
\begin{tabu}{>{\centering\setlength\hsize{2.2cm}}X |[1pt]l}

\textit{Category}~~1& 
\begin{array}{@{\hspace{2ex}}>{\raggedright\arraybackslash$} p{6.8cm} <{$}  >{\raggedright\arraybackslash$} p{3cm} <{$} >{\raggedright\arraybackslash$} p{2cm} <{$} }
x = \textsf{nil}   & x = y   \\
 x \neq \textsf{nil}  &  x \neq y  \\
\textit{shape}(\vec{x})  & x.\pf = \textsf{nil}    \\  
x \in \textit{shape}\_\textsf{heaplet}(\vec{y}) &  x.\pf \neq \textsf{nil}   \\                                               
x \notin \textit{shape}\_\textsf{heaplet}(\vec{y})   & x.\pf = y  \\                                   
\textit{shape}\_\textsf{heaplet}(\vec{x})\cap\textit{shape}\_\textsf{heaplet}(\vec{y}) = \emptyset  & x.\pf \neq y  \\

\end{array}
\end{tabu}
\\[1.9cm]

\begin{tabu}{>{\centering\setlength\hsize{2.2cm}}X |[1pt]l} 
\textit{Category}~~2& 
\begin{array}{@{\hspace{2ex}}>{\raggedright\arraybackslash$} p{6.8cm} <{$}  >{\raggedright\arraybackslash$} p{2.5cm} <{$} >{\raggedright\arraybackslash$} p{2.5cm} <{$} }
i    \in \textit{shape}\_\textit{key}\_\textsf{set}(\vec{x})   & x.\df =    i  \\%
i \notin \textit{shape}\_\textit{key}\_\textsf{set}(\vec{x})   & x.\df \neq i  & \\   %
\textit{shape}\_\textit{key}\_\textsf{set}(\vec{x}) \leq_\textsf{set} \{i\}  & x.\df \leq i  \\ %
\textit{shape}\_\textit{key}\_\textsf{set}(\vec{x}) \geq_\textsf{set} \{i\}  &  x.\df \geq i  \\ %
\textit{shape}\_\textit{key}\_\textsf{set}(\vec{x}) \leq_\textsf{set} \{y.\df\}  & x.\df =    y.\df        \\
\textit{shape}\_\textit{key}\_\textsf{set}(\vec{x}) \geq_\textsf{set} \{y.\df\} & x.\df \neq y.\df    \\

\textit{shape}\_\textit{key}\_\textsf{set}(\vec{x}) = \textit{shape}\_\textit{key}\_\textsf{set}(\vec{y})     & x.\df \leq y.\df        \\                                          
\textit{shape}\_\textit{key}\_\textsf{set}(\vec{x}) \leq_\textsf{set} \textit{shape}\_\textit{key}\_\textsf{set}(\vec{y})  &  x.\df \geq y.\df     \\
\textit{shape}\_\textit{key}\_\textsf{set}(\vec{x}) \geq_\textsf{set} \textit{shape}\_\textit{key}\_\textsf{set}(\vec{y}) \\
\textit{shape}\_\textit{key}\_\textsf{set}(\vec{x}) = \textit{shape}\_\textit{key}\_\textsf{set}(\vec{y}) &  \\[-0.3em]
\hspace*{19ex} ~~\cup~~\textit{shape}\_\textit{key}\_\textsf{set}(\vec{z}) & 

\end{array}
\end{tabu}

\\[3.2cm]

\begin{tabu}{>{\centering\setlength\hsize{2.2cm}}X |[1pt]l}
\textit{Category}~~3& 
\begin{array}{@{\hspace{2ex}}>{\raggedright\arraybackslash$} p{6.8cm} <{$}  >{\raggedright\arraybackslash$} p{5cm} <{$} >{\raggedright\arraybackslash$} p{3cm} <{$} } 
\textit{size}(\vec{x}) = i - j & \textit{size}(\vec{x}) = i      \\
\textit{size}(\vec{x}) - \textit{size}(\vec{y}) = i  & \textit{size}(\vec{x}) \leq i      \\
\textit{size}(\vec{x}) - \textit{size}(\vec{y}) = i - j &  \textit{size}(\vec{x}) \geq i   \\ 
\end{array}
\end{tabu}
\end{array}
\end{displaymath}

\caption{Templates for generating predicates. The operator $\leq_\textsf{set}$ denotes comparison between integer sets, where $A \leq_\textsf{set} B$ if and only if $\forall x \in A . \forall y \in B .~ x \leq y$.
The operator $\geq_\textsf{set}$ is similarly defined.
Shape properties such as \textsf{LinkedList}, \textsf{AVLtree}, etc., are recursively defined in \Dryad, separately,
and is extensible to any class of \Dryad defined shapes. Similarly, the definitions related to keys stored in a datastructure and the sizes of datastructures also stem from recursive definitions of them in \Dryad{}.
} \label{fig:predicates}

\end{figure}

The predicates are grouped into three categories, roughly in increasing complexity. Category~1 predicates involve shape-related properties, Category~2 involves properties related to the keys stored in the data-structure, and Category~3 predicates involve size-predicates on data structures (lengths of lists and heights of trees). 
Given a program to verify and its annotations, we choose the category of predicates depending on whether the specification refers to shape only, shapes and keys, or shapes, keys, and sizes (choosing a category includes the predicates of lower category as well). Then, predicates are automatically generated by instantiating the templates with all (combinations of) program variables. This approach allows for a fine-grained control over the predicates that are generated for a specific program and prevents the set of predicates from growing too large.

\subsubsection*{Evaluation}
\label{sec:evaluation}

We have evaluated our prototype on ten benchmark suits (82 routines in total) that contain standard algorithms on dynamic data structures, such as searching, inserting, or deleting items in lists and trees. These benchmarks were taken from the following sources:
\begin{enumerate*}[label={(\arabic*)}]
	\item	GNU C Library(glibc) singly/sorted linked lists,
	\item GNU C Library(glibc) doubly linked lists, 
	\item OpenBSD SysQueue,
	\item \textsc{GRASShopper}~\cite{DBLP:conf/cav/PiskacWZ13} singly linked lists,
	\item \textsc{GRASShopper}~\cite{DBLP:conf/cav/PiskacWZ13} doubly linked lists,
	\item \textsc{GRASShopper}~\cite{DBLP:conf/cav/PiskacWZ13} sorted linked lists,
	\item \textsc{VCDryad}~\cite{DBLP:conf/pldi/PekQM14} sorted linked lists,
	\item \textsc{VCDryad}~\cite{DBLP:conf/pldi/PekQM14} binary search trees, AVL trees, and treaps,
	\item AFWP~\cite{DBLP:conf/cav/ItzhakyBINS13} singly/sorted linked lists, and
	\item ExpressOS~\cite{DBLP:conf/asplos/MaiPXKM13} MemoryRegion.
\end{enumerate*}
The specifications for these programs are generally checks for their full functional correctness, such as 
preserving or altering shapes of data structures, inserting or deleting keys, filtering or finding elements, and sortedness of elements.
The specifications hence involve  separation logic with arithmetic as well as recursive definitions
that compute numbers (like lengths and heights), data-aggregating recursive functions (such as multisets
of keys stored in data-structures), and complex combinations of these properties
(e.g., to specify binary search trees, AVL trees and treaps). 
All programs are annotated in \textsc{Dryad}, and checking validity
of the resulting verification conditions is undecidable.

To create our benchmarks, we first picked all programs that contained iterative loops, 
\emph{erased} the user-provided loop invariants, and used our framework to synthesize adequate inductive invariants (our tool can synthesize multiple invariants for a program).
We also selected some programs that were purely recursive, where the contract for the function
had been strengthened to make the verification succeed. 
We \emph{weakened} these contracts to only state the specification (typically by removing formulas in the post-conditions
of recursively-called functions) and introduced annotation holes instead.
The goal was to synthesize strengthenings of these contracts that allow proving the program correct.
We also chose five straight-line programs, deleted their post-conditions, and evaluated
whether we can learn post-conditions for them.
Since our conjunctive learner learns the strongest invariant expressible as a conjunct,
we can use our framework to synthesize post-conditions as well.

After removing annotations from the benchmarks, we automatically inserted appropriate predicates over which to build invariants and contracts as described above. For all benchmark suits, conjunctions of these predicates were sufficient to prove the program correct.

\paragraph{Experimental Results}
We performed all experiments in a virtual machine running Ubuntu 16.04.1 on a single core of an Intel Core i7-7820\,HK 2.9\,GHz CPU with 2\,GB memory.
The box plots in Figure~\ref{fig:results-VCDryad} summarize the results of this empirical evaluation aggregated by benchmark suite, specifically the time required to verify the programs, the number of base predicates, and the number iterations in the learning process (see Appendix~\ref{app:experimental-results-heap-benchmarks} for full details).
Each box in the diagrams shows the lower and upper quartile (left and right border of the box, respectively), the median (line within the box), as well as the minimum and maximum (left and right whisker, respectively).

\pgfplotsset{smallboxplot/.style={
	width=52.5mm, height=57.5mm,
	ytick={1,...,10},
	minor tick style={draw=none},
	cycle list name = smallboxcyclelist,
}}

\pgfplotscreateplotcyclelist{smallboxcyclelist}{
	draw=black, fill=black!7.5 \\
	draw=black!60, fill=black!7.5 \\
}

\begin{figure}[th]

\makebox[\textwidth][c]
{
	\begin{tikzpicture}
		\begin{axis}[
			smallboxplot,
			name = time_plot,
			xlabel = {Time in s},
			ylabel = {Benchmark suite},
			xmode=log,
			yticklabels = {1 \textit{(16/3)}, 2 \textit{(9/3)}, 3 \textit{(3/1)}, 4 \textit{(8/1)}, 5 \textit{(8/1)}, 6 \textit{(11/2)}, 7 \textit{(3/2)}, 8 \textit{(11/3)}, 9 \textit{(9/2)}, 10 \textit{(4/1)}},
			]
			
			\addplot+[boxplot prepared={
				lower whisker=0.7, lower quartile=1.175,
				median=5.5,
				upper quartile=55.675, upper whisker=556.1,
				},] coordinates {};

			\addplot+[boxplot prepared={
				lower whisker=0.3, lower quartile=0.9,
				median=1.1,
				upper quartile=6.1, upper whisker=20.2,
				},] coordinates {};

			\addplot+[boxplot prepared={
				lower whisker=10.2, lower quartile=10.5,
				median=10.8,
				upper quartile=13.45, upper whisker=16.1,
				},] coordinates {};

			\addplot+[boxplot prepared={
				lower whisker=0.2, lower quartile=0.675,
				median=1.1,
				upper quartile=2.1, upper whisker=5.9,
				},] coordinates {};

			\addplot+[boxplot prepared={
				lower whisker=0.2, lower quartile=0.675,
				median=1.3,
				upper quartile=3.475, upper whisker=7.1,
				},] coordinates {};

			\addplot+[boxplot prepared={
				lower whisker=0.9, lower quartile=14.35,
				median=69.7,
				upper quartile=138.55, upper whisker=1679.7,
				},] coordinates {};

			\addplot+[boxplot prepared={
				lower whisker=1.1, lower quartile=4.45,
				median=7.8,
				upper quartile=55.5, upper whisker=103.2,
				},] coordinates {};

			\addplot+[boxplot prepared={
				lower whisker=0.2, lower quartile=0.2,
				median=0.6,
				upper quartile=141.1, upper whisker=599,
				},] coordinates {};

			\addplot+[boxplot prepared={
				lower whisker=0.1, lower quartile=1.8,
				median=3,
				upper quartile=9.6, upper whisker=339.3,
				},] coordinates {};

			\addplot+[boxplot prepared={
				lower whisker=0.1, lower quartile=0.175,
				median=0.25,
				upper quartile=1.625, upper whisker=5.6,
				},] coordinates {};
				
			\end{axis}
			
		\begin{axis}[
			smallboxplot,
			yticklabels={},
			name = predicate_plot,
			at={(time_plot.south east)}, xshift=2.5mm,
			xlabel = {Base predicates},
			xmode=log,
			]

			\addplot+[boxplot prepared={
				lower whisker=22, lower quartile=79.5,
				median=199.5,
				upper quartile=371, upper whisker=795,
				},] coordinates {};

			\addplot+[boxplot prepared={
				lower whisker=22, lower quartile=48,
				median=88,
				upper quartile=237, upper whisker=342,
				},] coordinates {};

			\addplot+[boxplot prepared={
				lower whisker=5, lower quartile=5,
				median=5,
				upper quartile=11.5, upper whisker=18,
				},] coordinates {};

			\addplot+[boxplot prepared={
				lower whisker=22, lower quartile=22,
				median=63,
				upper quartile=63, upper whisker=140,
				},] coordinates {};

			\addplot+[boxplot prepared={
				lower whisker=22, lower quartile=52.75,
				median=63,
				upper quartile=82.25, upper whisker=140,
				},] coordinates {};

			\addplot+[boxplot prepared={
				lower whisker=40, lower quartile=127.5,
				median=153,
				upper quartile=496, upper whisker=496,
				},] coordinates {};

			\addplot+[boxplot prepared={
				lower whisker=40, lower quartile=71,
				median=102,
				upper quartile=151.5, upper whisker=201,
				},] coordinates {};

			\addplot+[boxplot prepared={
				lower whisker=9, lower quartile=14,
				median=25,
				upper quartile=70, upper whisker=80,
				},] coordinates {};

			\addplot+[boxplot prepared={
				lower whisker=5, lower quartile=63,
				median=63,
				upper quartile=201, upper whisker=416,
				},] coordinates {};
			
			\addplot+[boxplot prepared={
				lower whisker=7, lower quartile=19.75,
				median=24,
				upper quartile=30.75, upper whisker=51,
				},] coordinates {};

			\end{axis}
			
		\begin{axis}[
				smallboxplot,
				yticklabels={},
				name = invariant_plot,
				xlabel = {Iterations},
				at={(predicate_plot.south east)}, xshift=2.5mm,
				xmode=log,
				]

			\addplot+[boxplot prepared={
				lower whisker=15, lower quartile=24.5,
				median=70.5,
				upper quartile=125.25, upper whisker=279,
				},] coordinates {};

			\addplot+[boxplot prepared={
				lower whisker=15, lower quartile=18,
				median=26,
				upper quartile=68, upper whisker=99,
				},] coordinates {};

			\addplot+[boxplot prepared={
				lower whisker=3, lower quartile=3,
				median=3,
				upper quartile=5, upper whisker=7,
				},] coordinates {};

			\addplot+[boxplot prepared={
				lower whisker=14, lower quartile=15,
				median=28.5,
				upper quartile=33, upper whisker=63,
				},] coordinates {};

			\addplot+[boxplot prepared={
				lower whisker=14, lower quartile=23.25,
				median=31.5,
				upper quartile=38, upper whisker=63,
				},] coordinates {};

			\addplot+[boxplot prepared={
				lower whisker=16, lower quartile=33,
				median=53,
				upper quartile=116.5, upper whisker=165,
				},] coordinates {};

			\addplot+[boxplot prepared={
				lower whisker=19, lower quartile=23.5,
				median=28,
				upper quartile=43.5, upper whisker=59,
				},] coordinates {};

			\addplot+[boxplot prepared={
				lower whisker=5, lower quartile=6,
				median=6,
				upper quartile=16.5, upper whisker=28,
				},] coordinates {};

			\addplot+[boxplot prepared={
				lower whisker=5, lower quartile=34,
				median=36,
				upper quartile=65, upper whisker=106,
				},] coordinates {};

			\addplot+[boxplot prepared={
				lower whisker=5, lower quartile=8,
				median=12.5,
				upper quartile=16, upper whisker=16,
				},] coordinates {};
				
			\end{axis}
			
	\end{tikzpicture}
	
}

	\vskip -.25\baselineskip
	\caption{Experimental evaluation of our prototype. The numbers in italic brackets shows the total number or programs in the suite (first number) and the maximum predicate category used (second number).}
	\label{fig:results-VCDryad}

\end{figure}
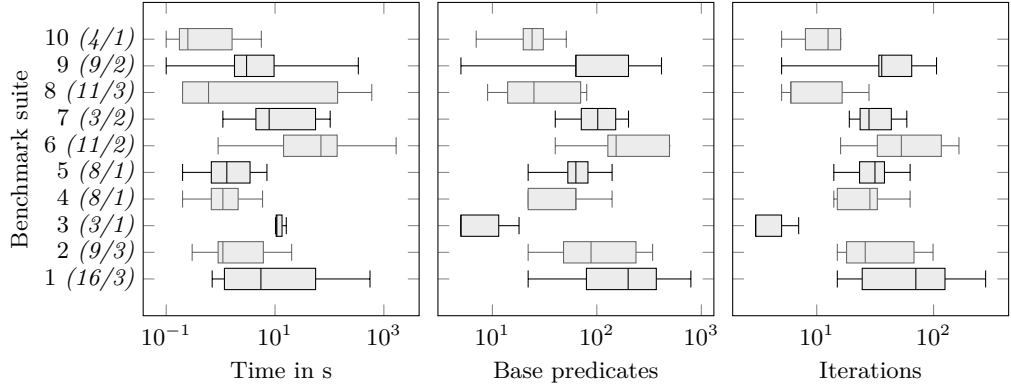

Our prototype was successful in learning invariants and contracts for all 82 benchmarks. 
Moreover, the fact that the median time for a great majority of benchmarks suits is less than 10\,s shows that our technique is extremely effective in finding inductive \textsc{Dryad} invariants.
We also observe that despite many examples having hundreds of base predicates, which in turn suggests a worst-case complexity of hundreds of iterations, the learner was able to learn with much fewer iterations and the number of predicates in the final invariant is small. 
This shows that the non-provability information provided by the natural proof engine provides much more information than what the worst-case suggests.

To the best of our knowledge, our prototype is the only tool currently able of fully automatically verifying this challenging benchmark set.
We must emphasize, however, that there are subsets of our benchmarks that can be solved by reformulating verification in decidable fragments of separation logic studied---we refer the reader to the related work in Section~\ref{sec:introduction} for a survey of such work. 
Our goal in this evaluation, however, is not to compete with other, mature tools on a subset of benchmarks, but to measure the efficacy of our proposed CD-NPI based invariant synthesis framework on the whole benchmark set.


\section{Application: Learning Invariants in the Presence of Bounded Quantifier Instantiation} \label{sec:quantified_fo}
Software verification of numerous applications must deal with quantification. 
For instance, quantifiers are often needed for axiomatizing theories that are not already equipped with decision procedures, for specifying properties of unbounded data structures and dynamically allocated memory, as well as for defining recursive properties of programs. 
For instance, the power of two function can be defined recursively using quantifiers as
\[ \mathit{pow2}(0) = 1 ~\text{and}~ \forall n \in \mathbb N \colon n > 0 \Rightarrow \mathit{pow2}(n) = 2 \cdot \mathit{pow2}(n-1). \]

Despite the fact that various important first-order theories are undecidable (e.g., the first-order theory of arithmetic with uninterpreted functions), modern SMT solvers implement a host of heuristics to cope with quantifier reasoning. Quantifier instantiation, including pattern-based quantifier instantiation (e.g., E-matching~\cite{DBLP:journals/jacm/DetlefsNS05}) and model-based quantifier instantiation~\cite{DBLP:conf/cav/GeM09}, are particularly effective heuristics in this context. The key idea of instantiation-based heuristics is to instantiate universally quantified formulas with a finite number of ground terms and then check for validity of the resulting quantifier-free formulas (whose theory needs to be decidable). The exact instantiation of ground terms varies from method to method, but most instantiation-based heuristics are necessarily incomplete in general due to the undecidability of the underlying decision problems.

We can apply invariant synthesis framework for verification engines that employ quantifier instantiation in the following way. Assume that $\calU$ is an undecidable first-order theory allowing uninterpreted functions and that $\calD$ is its decidable quantifier-free fragment. Then, quantifier instantiation can be seen as a transformation of a $\calU$-formula $\varphi$ (potentially containing quantifiers) into a $\calD$-formula $\apprx(\varphi)$ in which all existential quantifiers have been eliminated (e.g., using skolemization) and all universal quantifiers have been replaced by finite conjunctions over ground terms.\kern-.06em\footnote{Quantifier instantiation is usually performed iteratively, but we here abstract away from this fact.}
This means that if the $\calD$-formula $\apprx(\varphi)$ is valid, then the $\calU$-formula $\varphi$ is valid as well.
On the other hand, if $\apprx(\varphi)$ is not valid, one cannot deduce the validity of $\varphi$.
However, a $\calD$-model of $\apprx(\varphi)$ can be used to derive non-provability information as described in
Section~\ref{sec:CDNPI}.

We have implemented our learning framework for synthesizing invariants based on bounded quantifier instantiation.
Our prototype is based on \textsc{Boogie}/Z3 as the verification engine and uses \textsc{Houdini} to learn conjunctive invariants.
In the remainder of this section, we present empirical results of this implementation on benchmarks taken from competitions and verified systems such as IronFleet~\cite{DBLP:conf/sosp/HawblitzelHKLPR15}.

\subsubsection*{Evaluation}
We collected a benchmarks suite of twelve programs, which we obtained by simplifying programs found in
IronFleet~\cite{DBLP:conf/sosp/HawblitzelHKLPR15} (provably correct distributed systems),
\textsc{VSComp} (Verified Software Competition) benchmarks~\cite{vscomp2010}, ExpressOS~\cite{DBLP:conf/asplos/MaiPXKM13} (a secure operating system for mobile devices),
and sparse matrix multiplication programs~\cite{Buluc:2009:PSM:1583991.1584053}.
In these programs, quantifiers are used in specifying recursively defined predicates such as $\mathsf{power}(n,m)$ and $\mathsf{sum}(n)$, 
and various array properties such as minimum/maximum elements, existence of specific elements, no duplicate elements, permutations of array elements, relations between two arrays,
periodic properties of array elements, and bijective (injective and surjective) maps.
The specifications hence are undecidable and fall outside of the decidable array property fragment~\cite{Bradley:2006:WDA:2146228.2146256}. 
In particular, the array specifications involve strict comparison ($<$) between universally quantified index variables, array accesses in the index guard,
nested array accesses (e.g., $a_1[a_2[i]]$), arithmetic expressions over universally quantified index variables, and alternation of universal and existential quantifiers.

From this benchmark suite, we erased the user-defined loop invariants and used our framework to find adequate inductive invariants.
We generated a set of predicates that serve as the building blocks of our invariants. To this end, we used the pre-/post-conditions of the program being verified 
as templates from which the actual predicates are generated; the templates are instantiated using all combinations of program variables that occur in the program.
We also generated predicates for octagonal constraints, (i.e., relations between two integer variables of the form, $\pm x \pm y \le c$). For a few programs,
we also generated the octagonal predicates over array access expressions that appear in the program.

\paragraph{Experimental Results}
We performed all experiments in a virtual machine running Ubuntu 16.04.1 on a single core of an Intel Core i7-7820\,HK 2.9\,GHz CPU with 2\,GB memory. The results of these experiments are listed in Table~\ref{tbl:results-quantifier-instantiation}.

\begin{table}[t!h]
	\caption{Experimental results of the quantifier instantiation benchmarks. The column \emph{$|\calP|$} refer to the number of candidate predicates, the column \emph{\# Iterations} to the number of iterations of the teacher and learner, and the column \emph{$|\mathrm{Inv}|$} to the number of predicates in the inferred invariant.} \label{tbl:results-quantifier-instantiation}

	\centering
	\begin{tabular}{l@{\hskip 1.5em}*{4}{@{\hskip 1.5em}r}}
		\toprule
		Program & $|\calP|$ & \# Iterations & $|\mathrm{Inv}|$ & Time in s \\
		\midrule
		inverse & 414 & 126 & 73 & 9.04 \\
		power2 & 109 & 55 & 34 &  2.10 \\
		powerN & 160 & 60 & 31 & 13.52 \\
		recordArraySplit & 1264 & 49 & 51 & 57.46 \\
		recordArrayUnzip & 222 & 17 & 25 & 0.84 \\
		removeDuplicates & 280 & 67 & 86 & 4.43 \\
		setFind & 492 & 74 & 136 & 2.76 \\
		setInsert & 556 & 73 & 188 & 6.70 \\
		sparseMatrixGen & 816 & 278 & 90 & 22.07 \\
		sparseMatrixMul & 768 & 313 & 91 & 14.49 \\
		sum & 128 & 40 & 22 & 1.02 \\
		sumMax & 192 & 61 & 45 & 4.31 \\
		\bottomrule
	\end{tabular}
\end{table}

As can be seen from the table, our framework is effective in finding inductive invariants that result in proving the programs correct (with an average of less than a minute per routine).
Despite having hundreds of candidate predicates in many examples, which in turn suggests a worst-case complexity of hundreds of rounds, the learner was able
to learn with much fewer rounds. Again, the non-provability information provided by the verification engine provides much more information than the worst-case
suggests. 


\section{Conclusion} \label{sec:conclusion}
We have presented learning-based framework for invariant synthesis in the presence of sound but incomplete verification engines.
To prove that our technique is effective in practice, we have implemented our framework for two types of specifications: an expressive and undecidable dialect of separation logic called \Dryad for specifying heap properties and specifications involving universal quantification.
In both cases, our prototype turned out to be extremely effective in learning inductive invariants and pre/post-conditions. In particular, the benchmark suite for \Dryad-annotated programs is extremely challenging, containing an extensive list of standard algorithms on dynamic data structures, and we are not aware of any other technique that can handle this benchmark suite.

Several future research directions are interesting. First, the framework we have developed is based on
CEGIS where the invariant synthesizer synthesizes invariants using non-provability information
but does not directly work on the program's structure. It would be interesting to extend
white-box invariant generation techniques such as interpolation/IC3/PDR, 
working using $\calD$ (or $\calB$) abstractions of the
program directly in order to synthesize invariants for them. 
Second, in the NPI learning framework we have put forth, it would be interesting to change the underlying
logic of communication $\calB$ to a richer logic, say the theory of arithmetic and uninterpreted functions.
The challenge here would be to extract non-provability information from the models to the richer theory,
and pairing them with synthesis engines that synthesize expressions against constraints in $B$.
Finally, we think invariant learning should also include \emph{experience} gained in verifying other programs in the past, both manually and automatically. A learning algorithm that combines logic-based synthesis with experience gained from repositories of verified programs can be more effective.

\bibliographystyle{splncs03}
\bibliography{main}

\clearpage
\appendix

\section{Detailed Results of the Heap Invariants Benchmarks}
\label{app:experimental-results-heap-benchmarks}

\begin{longtable}{l@{\hskip 1.5em}*{5}{@{\hskip 1.5em}r}}
	\caption{Experimental results of the heap invariants benchmarks. The column \emph{$|\calP|$} refer to the number of candidate predicates, the column \emph{Cat.}\ corresponds to the category of predicates used, the column \emph{\# Iterations} to the number of iterations of the teacher and learner, and the column \emph{$|\mathrm{Inv}|$} to the number of predicates in the inferred invariant. A \textsuperscript{\textdagger} indicates contract strengthening, while a \textsuperscript{\textasteriskcentered} indicates post condition learning. \label{tab:exp}	}\\
	\endfirsthead
	\caption[]{continued}
	\endhead

	\multicolumn{6}{c}{(1) GNU C Library(glibc) Singly and Sorted Linked-List} \\[.5ex]
	\toprule
	Program & $|\calP|$ & Cat. & \# Iterations & $|\mathrm{Inv}|$ & Time in s \\
	\midrule
	g\_slist\_copy & 368 & 2 & 123 & 101 & 55 \\
	g\_slist\_find & 48 & 2 & 18 & 9 & 0.8 \\
	g\_slist\_free & 22 & 1 & 15 & 1 & 1.2 \\
	g\_slist\_index & 237 & 3 & 68 & 57 & 6.3 \\
	g\_slist\_insert & 464 & 2 & 160 & 50 & 219.1 \\
	g\_slist\_insert\_before & 795 & 2 & 279 & 114 & 556.1 \\
	g\_slist\_insert\_sorted & 520 & 2 & 193 & 135 & 210.6 \\
	g\_slist\_last & 32 & 2 & 19 & 8 & 0.7 \\
	g\_slist\_length & 54 & 3 & 20 & 12 & 0.9 \\
	g\_slist\_nth & 88 & 3 & 26 & 17 & 1.1 \\
	g\_slist\_nth\_data & 342 & 3 & 99 & 62 & 9.2 \\
	g\_slist\_position & 162 & 3 & 32 & 18 & 2.7 \\
	g\_slist\_remove & 140 & 1 & 73 & 28 & 4.7 \\
	g\_slist\_remove\_all & 380 & 2 & 132 & 15 & 57.7 \\
	g\_slist\_remove\_link & 325 & 1 & 85 & 57 & 14.3 \\
	g\_slist\_reverse & 117 & 2 & 58 & 6 & 4.5 \\
	\bottomrule \\[.5ex]
	\multicolumn{6}{c}{(2) GNU C Library(glibc) Doubly Linked-List} \\[.5ex]
	\toprule
	Program & $|\calP|$ & Cat. & \# Iterations & $|\mathrm{Inv}|$ & Time in s \\
	\midrule
	g\_list\_find & 48 & 2 & 18 & 9 & 0.8 \\
	g\_list\_free & 22 & 1 & 15 & 1 & 0.9 \\
	g\_list\_index & 237 & 3 & 68 & 57 & 6.1 \\
	g\_list\_last & 22 & 1 & 15 & 6 & 0.3 \\
	g\_list\_length & 88 & 3 & 24 & 16 & 1.1 \\
	g\_list\_nth & 88 & 3 & 26 & 17 & 1.1 \\
	g\_list\_nth\_data & 342 & 3 & 99 & 62 & 9.4 \\
	g\_list\_position & 162 & 3 & 32 & 18 & 2.9 \\
	g\_list\_reverse & 320 & 2 & 84 & 2 & 20.2 \\
	\bottomrule \\[.5ex]
	\multicolumn{6}{c}{(3) OpenBSD SysQueue} \\[.5ex]
	\toprule
	Program & $|\calP|$ & Cat. & \# Iterations & $|\mathrm{Inv}|$ & Time in s \\
	\midrule
	squeue\_insert\_head\textsuperscript{\textasteriskcentered} & 5 & 1 & 3 & 2 & 10.8 \\
	squeue\_insert\_tail\textsuperscript{\textasteriskcentered} & 5 & 1 & 3 & 3 & 16.1 \\
	squeue\_remove\_head\textsuperscript{\textasteriskcentered} & 18 & 1 & 7 & 5 & 10.2 \\
	\bottomrule \pagebreak
	%
	%
	\multicolumn{6}{c}{(4) GRASShopper~\cite{DBLP:conf/cav/PiskacWZ13} Singly Linked-List} \\[.4ex]
	\toprule
	Program & $|\calP|$ & Cat. & \# Iterations & $|\mathrm{Inv}|$ & Time in s \\
	\midrule
	sl\_concat & 63 & 1 & 26 & 15 & 0.9 \\
	sl\_copy & 63 & 1 & 32 & 12 & 2.7 \\
	sl\_dispose & 22 & 1 & 14 & 1 & 0.7 \\
	sl\_filter & 140 & 1 & 63 & 9 & 5.9 \\
	sl\_insert & 63 & 1 & 31 & 19 & 1.3 \\
	sl\_remove & 22 & 1 & 15 & 6 & 0.6 \\
	sl\_reverse & 63 & 1 & 36 & 4 & 1.9 \\
	sl\_traverse & 22 & 1 & 15 & 4 & 0.2 \\
	\bottomrule \\[.5ex]
	\multicolumn{6}{c}{(5) GRASShopper~\cite{DBLP:conf/cav/PiskacWZ13} Doubly Linked-List} \\[.4ex]
	\toprule
	Program & $|\calP|$ & Cat. & \# Iterations & $|\mathrm{Inv}|$ & Time in s \\
	\midrule
	dl\_concat & 63 & 1 & 26 & 15 & 0.7 \\
	dl\_copy & 63 & 1 & 32 & 12 & 3.1 \\
	dl\_dispose & 140 & 1 & 44 & 4 & 7.1 \\
	dl\_filter & 140 & 1 & 63 & 9 & 4.6 \\
	dl\_insert & 63 & 1 & 31 & 19 & 0.9 \\
	dl\_remove & 22 & 1 & 15 & 6 & 0.4 \\
	dl\_reverse & 63 & 1 & 36 & 4 & 1.7 \\
	dl\_traverse & 22 & 1 & 14 & 4 & 0.2 \\
	\bottomrule \\[.5ex]
	\multicolumn{6}{c}{(6) GRASShopper~\cite{DBLP:conf/cav/PiskacWZ13} Sorted Linked-List} \\[.4ex]
	\toprule
	Program & $|\calP|$ & Cat. & \# Iterations & $|\mathrm{Inv}|$ & Time in s \\
	\midrule
	sls\_concat & 153 & 2 & 38 & 27 & 11 \\
	sls\_copy & 496 & 2 & 144 & 94 & 1679.7 \\
	sls\_dispose & 40 & 2 & 16 & 3 & 1.6 \\
	sls\_double\_all & 496 & 2 & 106 & 102 & 118.8 \\
	sls\_filter & 496 & 2 & 127 & 30 & 158.3 \\
	sls\_insert & 153 & 2 & 53 & 29 & 17.7 \\
	sls\_merge & 416 & 2 & 63 & 29 & 327.0 \\
	sls\_remove & 496 & 2 & 165 & 119 & 69.7 \\
	sls\_reverse & 102 & 2 & 28 & 13 & 21.6 \\
	sls\_split & 153 & 2 & 53 & 29 & 98.1 \\
	sls\_traverse & 40 & 2 & 18 & 9 & 0.9 \\
	\bottomrule \\[.5ex]
	\multicolumn{6}{c}{(7) VCDryad~\cite{DBLP:conf/pldi/PekQM14} Sorted Linked-List} \\[.5ex]
	\toprule
	Program & $|\calP|$ & Cat. & \# Iterations & $|\mathrm{Inv}|$ & Time in s \\
	\midrule
	find\_last\_sorted & 40 & 2 & 19 & 11 & 1.1 \\
	reverse\_sorted & 102 & 2 & 28 & 13 & 7.8 \\
	sorted\_insert\_iter & 201 & 2 & 59 & 48 & 103.2 \\
	\bottomrule \pagebreak
	%
	%
	\multicolumn{6}{c}{(8) VCDryad~\cite{DBLP:conf/pldi/PekQM14} Trees} \\[.5ex]
	\toprule
	Program & $|\calP|$ & Cat. & \# Iterations & $|\mathrm{Inv}|$ & Time in s \\
	\midrule
	avl-delete-rec\textsuperscript{\textdagger} & 72 & 3 & 16 & 5 & 449.1 \\
	avl-find-smallest\textsuperscript{\textdagger} & 19 & 3 & 5 & 11 & 0.2 \\
	avl-insert-rec\textsuperscript{\textdagger} & 72 & 3 & 23 & 14 & 102.2 \\
	bst-delete-rec\textsuperscript{\textdagger} & 68 & 2 & 16 & 11 & 180 \\
	bst-find-rec\textsuperscript{\textdagger} & 23 & 2 & 6 & 9 & 0.5 \\
	bst-insert-rec\textsuperscript{\textdagger} & 68 & 2 & 28 & 16 & 64.8 \\
	traverse-inorder\textsuperscript{\textdagger} & 9 & 3 & 6 & 3 & 0.2 \\
	traverse-posttorder\textsuperscript{\textdagger} & 9 & 3 & 6 & 3 & 0.2 \\
	traverse-preorder\textsuperscript{\textdagger} & 9 & 3 & 6 & 3 & 0.2 \\
	treap-delete-rec\textsuperscript{\textdagger} & 80 & 3 & 17 & 13 & 599.0 \\
	treap-find-rec\textsuperscript{\textdagger} & 25 & 3 & 6 & 11 & 0.6 \\
	\bottomrule \\[.5ex]
	\multicolumn{6}{c}{(9) AFWP~\cite{DBLP:conf/cav/ItzhakyBINS13} Singly and Sorted Linked-List} \\[.4ex]
	\toprule
	Program & $|\calP|$ & Cat. & \# Iterations & $|\mathrm{Inv}|$ & Time in s \\
	\midrule
	SLL-create & 5 & 1 & 5 & 1 & 0.1 \\
	SLL-delete-all & 22 & 1 & 14 & 1 & 5.3 \\
	SLL-delete & 265 & 1 & 106 & 47 & 9.6 \\
	SLL-filter & 63 & 1 & 34 & 9 & 2.1 \\
	SLL-find & 140 & 1 & 53 & 45 & 3 \\
	SLL-insert & 201 & 2 & 65 & 26 & 45.6 \\
	SLL-last & 63 & 1 & 34 & 9 & 1.2 \\
	SLL-merge & 416 & 2 & 71 & 46 & 339.3 \\
	SLL-reverse & 63 & 1 & 36 & 4 & 1.8 \\
	\bottomrule \\[.5ex]
	\multicolumn{6}{c}{(10) ExpressOS~\cite{DBLP:conf/asplos/MaiPXKM13} MemoryRegion} \\[.5ex]
	\toprule
	Program & $|\calP|$ & Cat. & \# Iterations & $|\mathrm{Inv}|$ & Time in s \\
	\midrule
	memory\_region\_find & 24 & 1 & 16 & 2 & 0.2 \\
	memory\_region\_init\textsuperscript{\textasteriskcentered} & 7 & 1 & 5 & 4 & 0.1 \\
	memory\_region\_insert & 51 & 1 & 16 & 3 & 0.3 \\
	split\_memory\_region\textsuperscript{\textasteriskcentered} & 24 & 1 & 9 & 6 & 5.6 \\
	\bottomrule
\end{longtable}

\end{document}